\documentclass[10pt,reqno]{amsart}
\usepackage{amsthm,amsfonts,amsmath,latexsym,amssymb,bbold,dsfont}
\usepackage{amssymb,latexsym,amsfonts,amsmath,amsthm,mathabx,stmaryrd}
\usepackage[unicode]{hyperref} %new try
\usepackage{verbatim}
\usepackage[integrals]{wasysym}

\usepackage[dvips]{graphics}
\usepackage{graphicx,xcolor}
\usepackage{epsfig}

\usepackage[rightcaption]{sidecap}

\usepackage{soul}    % for striketrough -- for draft only

\usepackage{mathtools}
\usepackage{mfirstuc}
\usepackage[normalem]{ulem}
\usepackage[utf8]{inputenc}
\usepackage{cite}
\def\a{\alpha}
\def\b{\beta}
\def\g{\gamma}
\def\d{\delta}
\def\dd{\mathrm{d}}

\def\eps{\varepsilon}
\def\s{\sigma}

\def\D{\Delta}

\def\l{\lambda}
\def\L{\Lambda}

\def\C{\mathbb C}

\def\R{{\mathbb R}}

\def\N{\mathbb N}
\def\tr{{\rm tr}\,}

\def \Lp {\mathsf L}  % L^p space L
\newcommand{\NH}{Nevanlinna--Herglotz}

\newcommand{\1}
{{\,\vrule depth3pt height9pt}{\vrule depth3pt height9pt}
{\vrule depth3pt height9pt}{\vrule depth3pt height9pt}\,}

\usepackage{xparse,aliascnt,hyperref,bookmark}

\NewDocumentCommand{\xnewtheorem}{m o m}
 {%
  \IfNoValueTF{#2}
   {\newtheorem{#1}{#3}}
   {%
    \newaliascnt{#1}{#2}%
    \newtheorem{#1}[#1]{#3}%
    \aliascntresetthe{#1}%
    \expandafter\newcommand\csname #1autorefname\endcsname{\makefirstuc  {\lowercase {#3}}}%
   }%
 }

\hypersetup{
    colorlinks,
    linkcolor={red!50!black},
    citecolor={blue!50!black},
    urlcolor={blue!80!black}
}

\theoremstyle{plain}
\newtheorem{thmc}{ERROR}[section]
\iftrue %CAPS
\xnewtheorem{thm}[thmc]{THEOREM}
\xnewtheorem{lemma}[thmc]{LEMMA}
\xnewtheorem{cl}[thmc]{COROLLARY}
\xnewtheorem{prop}[thmc]{PROPOSITION}
\xnewtheorem{assumptions}[thmc]{ASSUMPTIONS}
\xnewtheorem{conj}[thmc]{CONJECTURE}
\xnewtheorem{fact}[thmc]{FACT}
\theoremstyle{definition}
\xnewtheorem{defin}[thmc]{DEFINITION}

\theoremstyle{remark}
\xnewtheorem{remark}[thmc]{REMARK}
\xnewtheorem{remarks}[thmc]{REMARKS}

\else

\theoremstyle{plain}
\xnewtheorem{thm}[thmc]{Theorem}
\xnewtheorem{lemma}[thmc]{Lemma}
\xnewtheorem{cl}[thmc]{Corollary}
\xnewtheorem{prop}[thmc]{Proposition}
\xnewtheorem{assumption}[thmc]{Assumption}
\xnewtheorem{conj}[thmc]{Conjecture}
\xnewtheorem{fact}[thmc]{Fact}

\theoremstyle{definition}
\xnewtheorem{defin}[thmc]{Definition}

\theoremstyle{remark}
\newtheorem{remark}[thmc]{Remark}
\newtheorem{remarks}[thmc]{Remarks}
\fi

\newcommand{\be}{\begin{equation}}
\newcommand{\ee}{\end{equation}}
\newcommand{\bea}{\begin{eqnarray}}
\newcommand{\eea}{\end{eqnarray}}
\newcommand{\beax}{\begin{eqnarray*}}
\newcommand{\eeax}{\end{eqnarray*}}

\newcommand{\calI}{{\mathcal{I}}}
\newcommand{\Cpm}{\Pi_{\pm}}
\newcommand{\Cp}{\Pi_{+}}
\newcommand{\Cm}{\Pi_{-}}
\newcommand{\Renyi}{R\'{e}nyi}

 %, asymptotic analysis}

\numberwithin{equation}{section}

\begin{document}

\title[Entanglement entropy of massless Dirac fermions in dimension one]{Entanglement entropy in the ground state of non-interacting massless Dirac fermions in dimension one}

\date{April 10, 2024}
%\tableofcontents
\author[F.~Ferro, P.~Pfeiffer, W.~Spitzer]{Fabrizio Ferro, Paul Pfeiffer and Wolfgang Spitzer}
\address{Fakult\"at f\"ur Mathematik und Informatik, FernUniversit\"at in Hagen, Universit\"atsstra\ss e 1, 58097 Hagen, Germany} \email{fabrizio.ferro@studium.fernuni-hagen.de}  \email{paul.pfeiffer@fernuni-hagen.de} \email{wolfgang.spitzer@fernuni-hagen.de}

\begin{abstract} We present a novel proof of a formula of Casini and Huerta for the entanglement entropy of the ground state of non-interacting massless Dirac fermions in dimension one localized to (a union of) intervals and generalize it to the case of R\'enyi entropies. At first, we  prove that these entropies are well-defined for non-intersecting intervals. This is accomplished by an inequality of Alexander V.~Sobolev. Then we compute this entropy using a trace formula for Wiener--Hopf operators by Harold Widom. For intersecting intervals, we discuss an extended entropy formula of Casini and Huerta and support this with a proof for polynomial test functions (instead of entropy). 
\end{abstract}

\maketitle
\centerline{\emph{In memory of Mary-Beth Ruskai (1944--2023)}}
\bigskip

\section{Introduction}

We consider non-interacting, massless relativistic Dirac fermions on the real line. For simplicity, we may assume that these fermions are spinless. The single-particle Hamiltonian of this Fermi system is given by the momentum operator $-\mathrm{i}\,\dd/\dd x$, which acts self-adjointly on the Sobolev space, $\mathsf{H}^1(\R)$. The (pure) ground state of fermions is characterized by the (orthogonal) projection operator $P\coloneqq 1_{\R^+}(-\mathrm{i}\,\dd/\dd x)$ on the Hilbert space $\mathsf{L}^2(\R)$ of square-integrable, complex-valued functions on $\R$ with inner product $\langle\cdot,\cdot\rangle$, or equivalently, $P$ is the spectral projection of $-\mathrm{i}\,\dd/\dd x$ onto the positive (infinite) Fermi sea, $\R^+\coloneqq \{\xi\in\R:\xi\ge0\}$. Below in \eqref{def: Op}, this operator $P$ will be denoted by $\mathrm{Op}(1_{\R^+})$. It is related to the Hilbert transform $\mathrm{H}$ on $\mathsf{L}^2(\R)$ (see \cite[p. 26]{Stein}), which is given by the convolution with $x\mapsto 1/(\pi x)$, that is, $(\mathrm{H}\varphi)(x)\coloneqq \pi^{-1}\int_\R \dd y\, \varphi(y)/(x-y)$ for $x\in\R$ and $\varphi\in\mathsf{L}^2(\R)$, interpreted as a principal-value integral. The simple relation between $P$ and $\mathrm{H}$ is 
\be P = \frac12 \mathds{1} +\frac{1}{2}\mathrm{i}\,\mathrm{H}\,,
\ee
where $\mathds{1}$ is the identity operator on $\mathsf{L}^2(\R)$. Therefore, $P$ has the integral kernel
\begin{align} \label{kernel P}
P(x,y) = \frac{1}{2}\d(x-y) - \frac{1}{2\pi\mathrm{i}} \frac{1}{x-y}\,,\quad x,y\in\R\,.
\end{align}
As a pure state, $P$ has absolute entropy zero. Given some (bounded or possibly unbounded) Borel subset $\Lambda\subset\R$ we reduce this ground state locally to $\Lambda$ and define on $\mathsf{L}^2(\R)$ the operator 
\be \label{restricted}
P(\Lambda)\coloneqq 1_\L {P} 1_\L\,,
\ee 
where we use the same symbol, $1_\L$, for the indicator function of $\L$ and for the multiplication operator by $1_\L$ in position space. To this end, let $X$ be the space multiplication operator informally defined as $(X \varphi)(x)\coloneqq x \varphi(x)$, $x\in\R$, for suitable $\varphi\in\mathsf{L}^2(\R)$. Then we identify $1_\Lambda$ with $1_{\Lambda}(X)$. More generally, let $f\in\mathsf{L}^\infty(\R)$ be a bounded, measurable function on $\R$, then the bounded operator $f(X)$ on $\mathsf{L}^2(\R)$ is defined, as usual, by
\be \label{def f(X)}
\big(f(X)\varphi\big)(x)\coloneqq f(x)\varphi(x)\,,\quad x\in\R\,,\varphi\in\mathsf{L}^2(\R)\,.
\ee
The operator $P(\L)$ in \eqref{restricted} characterizes the ground state reduced locally to $\L$ in the sense that all its expectation values are determined by the inner products $\langle f,P(\L)g\rangle$, for $f,g\in\mathsf{L}^2(\R)$, and the Wick rule or quasi-free property to compute higher order correlation functions in terms of a determinant. See, for example \cite[Chapter 4]{HLS11}, \cite{LX18,Xu21} for details. For general $\L$, the operator $P(\L)$ is not a projection and no longer corresponds to a pure state of fermions but it satisfies $0\le P(\L)\le \mathds{1}$. Informally, we denote the absolute (von-Neumann) entropy of the ground state localized to $\L$ by $S(\L) \coloneqq \tr h_1(P(\L))$, see \eqref{entropy von Neumann} for the definition of $h_1$. However, for general $\L\subset \R$, the operator $h_1(P(\L))$ is positive but not trace class and thus $S(\L)$ is infinite. Therefore,  we consider certain relative entropies or entropy differences. Let $\L'\subset\R$ be another Borel subset disjoint from $\L$. Our goal is to understand a quantity which is informally the sum of the individual entropies $S(\L)+S(\L')$ minus the entropy $S(\L\cup\L')$ of the union. Below we denote this by $\tr\D(\L,\L';h_1)$ with the von-Neumann entropy function $h_1$ defined in \eqref{entropy von Neumann}, where
\be \label{Delta}
\D(\L,\L';f) \coloneqq f\big(P(\L)\big) + f\big(P(\L')\big) - f\big(P(\L\cup \L')\big)\,.
\ee
It turns out that (under certain assumptions on $\L$, $\L'$ and $f$, see \autoref{assump}) $\D(\L,\L';  f )$ is a trace-class operator. For $f=h_1$, it is a measure of entanglement of the ground state localized to $\L$ and $\L'$. We refer to this as the von-Neumann entanglement entropy and denote this by $\mathsf{EE}(\L,\L')$. Another frequently used and maybe more appropriate name for this entanglement entropy is mutual information. Casini and Huerta computed this entropy difference explicitly (for the case of a union of intervals and the von-Neumann entropy) and therefore we refer to this as the Casini--Huerta formula. Later, using the same integral representation for the von-Neumann entropy function, Longo and Xu (rigorously) proved this formula. More recently, using the Cauchy formula and the so-called replica-trick, Blanco et al.~\cite{Blanco2022} have computed the integer R\'enyi entanglement entropies, that is, $\tr \D(\L,\L';h_n)$ for $n\in\N$ with the R\'enyi entropy function $h_n$ defined in \eqref{entropy}.

In our approach proving the trace-class property boils down to showing that $1_\L \mathrm{Op}(1_{\R^+}) 1_{\L'}$ lies in the Schatten--von Neumann class $\mathcal S_p$ for any $p>0$ and  employing an inequality by A.V.~Sobolev that reduces the entropy difference to the $p$ (quasi-)norm of $1_\L \mathrm{Op}(1_{\R^+})1_{\L'}$.

We are not aware of a (rigorous) generalization of the approach by Casini--Huerta (and the method of proof of Longo--Xu) to all R\'enyi entropies. We do mention our generalization to R\'enyi entropies with R\'enyi index $\a<1$, see \autoref{section: master thesis}. Another drawback of Casini--Huerta's method is that it relies on the explicit knowledge of the resolvent of $P(\Lambda)$, which is usually not available explicitly and the main reason why there is no such simple formula in the ground state of massive Dirac fermions. See \cite{Xu23} for the latest developments on massive free Dirac fermions in one spatial dimension. At this point we also mention the mathematical progress on entanglement entropy for free Dirac fermions in higher spatial dimensions in the recent papers \cite{BM23,FLS23}.

Let us now explain the set-up and the results by Harold Widom \cite{Widom82} on the trace of certain Wiener--Hopf operators. To this end, let (the symbol) $a\in \Lp^\infty(\R)$ be real-valued. Then, $\mathrm{Op}(a) = a(-\mathrm{i}\,\dd/\dd x)$ is the operator acting on Schwartz functions $\varphi$ on $\R$ as
\be \label{def: Op}
\big(\mathrm{Op}(a) \varphi\big)(x) \coloneqq \frac{1}{2\pi} \int_\R \mathrm{d}y \int_\R \mathrm{d}\xi\,\exp\big(\mathrm{i}\xi(x-y)\big) \, a(\xi) \varphi(y) = \mathcal F^{-1}\big(a\mathcal F(\varphi)\big)(x) \,,\quad x\in\R\,,
\ee
where $\mathcal F$ (and $\mathcal F^{-1}$) is the (inverse) Fourier transformation with the convention
\begin{align*}\hat{\varphi}(\xi)\coloneqq(\mathcal F\varphi)(\xi)&\coloneqq \int_\R \dd x\, \varphi(x) \, \mathrm{e}^{-\mathrm{i}\xi x}\,,\quad \xi\in\R\,,
\\
(\mathcal F^{-1}\varphi)(x)&\coloneqq \frac{1}{2\pi}\int_\R \dd \xi\, \varphi(\xi) \, \mathrm{e}^{\mathrm{i}\xi x}\,,\quad x\in\R\,.
\end{align*}
The operator $\mathrm{Op}(a)$ can be extended uniquely to a bounded operator on $\mathsf{L}^2(\R)$. In this sense, $P = \mathrm{Op}(1_{\R^+}) = 1_{\R^+}(-\mathrm{i}\,\dd/\dd x)$ is understood. Let
\begin{align} \label{WH} 
W(a) \coloneqq {1}_{\R^+} \mathrm{Op}(a) {1}_{\R^+}
\end{align}
be the (truncated) Wiener--Hopf operator on $\Lp^2(\R)$. It can also be considered as an operator on $\mathsf{L}^2(\R^+)$. Widom assumed that the symbol $a$ satisfies (see \cite[(2)]{Widom82})
\[ \int_\R \dd \xi \, |\xi| \big|(\mathcal F^{-1}a)(\xi)\big|^2 <\infty .
\]
Then, under the conditions on the (test) function $f$ that  (i) $f\in \Lp^1(\R)$ and (ii) $\xi\mapsto \xi^2 \hat{f}(\xi)\in \Lp^1(\R)$, Widom proved (see \cite[Theorem 1b]{Widom82}) that the operator 
\be \label{def D(a;f)} D(a;f) \coloneqq f\big(W(a)\big) - W\big(f(a)\big)
\ee
is trace class and its trace is given by the formula
\be \label{Widom}
\tr D(a;f) = \frac{1}{8\pi^2} \int_{\R\times\R} \mathrm{d}\xi_1 \mathrm{d}\xi_2\, \frac{U\big(a(\xi_1),a(\xi_2);f\big)}{(\xi_1-\xi_2)^2} \eqqcolon \mathcal B(a;f)\,,
\ee
where
\be \label{def U}
U(\s_1,\s_2;f) \coloneqq \int_0^1 \mathrm{d}t \, \frac{f\big((1-t)\s_1+t\s_2\big) - (1-t)f(\s_1)-tf(\s_2)}{t(1-t)}\,,
\quad \s_1,\s_2\in\R\,.
\ee
For a symbol $a\in\mathsf{L}^\infty(\R)$, Widom introduced the Besov-norm $\interleave\cdot\interleave$ with  
\be \interleave a \interleave^2 \coloneqq \mathcal B\big(a;t\mapsto t(1-t)\big) = \frac{1}{8\pi^2} \int_{\R^2} \dd\xi_1\dd\xi_2\,\frac{\big(a(\xi_1)-a(\xi_2)\big)^2}{(\xi_1-\xi_2)^2}
\ee
and the vector space $\mathcal K\subset \mathsf{L}^\infty(\R)$ of those symbols for which this norm is finite. Then, for two symbols $a,b\in\mathcal K$, he defined the (Hankel) operator
\be \label{Hankel}
H(a,b)\coloneqq W(ab) - W(a)W(b)
\ee
and proved that it is trace class with the trace-norm estimate
\be \label{Hankel estimate}
\|H(a,b)\|_{\mathcal S_1} \le \interleave a\interleave\, \interleave b\interleave\,.
\ee
An important technical relation (see \cite[(13)]{Widom82}) that we will use for $a\in\mathcal K$ and $x\in\R$ is
\be \label{exp W}
W\big(\exp(\mathrm{i}x a)\big) = {1}_{\R^+} \exp\big(\mathrm{i}xW(a)\big) + \mathrm{i}\int_0^x \dd y\, H\big(\exp(\mathrm{i}y a),a\big)\, \exp\big(\mathrm{i}(x-y)W(a)\big)\,.
\ee
Because of misprints in \cite[(13)]{Widom82} and for the convenience of the reader, we quickly recall here Widom's derivation. We start from 
\begin{align*} \frac{\dd }{\dd y} &\Big(W(\exp(\mathrm{i} y a)\big) \exp\big(-\mathrm{i} y W(a)\big) \Big)
\\
&= \mathrm{i} W\big(a\exp(\mathrm{i} y a)\big) \exp\big(-\mathrm{i} y W(a)\big) - \mathrm{i} W\big(\exp(\mathrm{i} y a)\big) W(a)\exp\big(-\mathrm{i} y W(a)\big)
\\
&=\mathrm{i} H\big(\exp(\mathrm{i} y a),a\big) \exp\big(-\mathrm{i} y W(a)\big)\,.
\end{align*}
Integrating $y$ from 0 to $x$, we obtain
\begin{align*} W\big(\exp(\mathrm{i} x a)\big) \exp\big(-\mathrm{i} x W(a)\big)- W(1) = \mathrm{i} \int_0^x \dd y\, H\big(\exp(\mathrm{i} y a),a\big) \exp\big(-\mathrm{i} y W(a)\big)\,.
\end{align*}
Finally, we multiply from the right by $\exp\big(+\mathrm{i} x W(a)\big)$ and use that $W(1) = 1_{\R^+}$. This yields \eqref{exp W}.

It is not difficult to see (cf.~\cite[(8)]{Widom82}) that in this case,
\be \label{1.13}
\interleave\exp(\mathrm{i}x a)\interleave \le |x| \,\interleave a\interleave\,.
\ee

All this is not directly applicable to the functions $f$ that we have in mind, namely the R\'enyi entropy function $h_\a$ in combination with the symbol $a=1_\L$ for a subset $\L\subset \R$. The R\'enyi entropy function $h_\a\colon \R \to \R$ with index $\alpha>0$ is defined for $t\in(0,1)$ as
\begin{align} \label{entropy} h_\a(t)&\coloneqq \frac{1}{1-\a} \ln\big[t^\a +(1-t)^\a\big] \,,\quad \mbox{ if }\alpha\not=1\,,
\\
h_1(t)&\coloneqq\lim_{\alpha'\to1}h_{\a'}(t) = -t\ln(t)-(1-t)\ln(1-t)\,,\quad \mbox{ if }\alpha=1\,. \label{entropy von Neumann}
\end{align}
For $t \in\R\setminus (0,1)$, we set $h_{\a}(t)\coloneqq 0$. 

Notice that $\mathcal B(1_\L;f)$ is, in general, not well-defined. For instance, if $f(t)=t(1-t)$, then $U(\s_1,\s_2;f) = (\s_1-\s_2)^2/(8\pi^2)$ and the integral in $\mathcal B(1_\L;f)$ does not converge. The same holds with $f=h_\a$. In \cite{Sobolev16}, conditions on $a$ and $f$ were studied so that this coefficient $\mathcal B(a;f)$ is well-defined. If we smoothen the symbol $1_\L$ and denote the smooth symbol by $\varphi_{\eps}$, then $\mathcal B(\varphi_\eps;h_\a)$ is indeed well-defined, see \cite[Theorem 3.2]{Sobolev16}.

Similarly to $\D(\L,\L';f)$, we also introduce the operator difference
\be\label{D_operator} 
D(\L,\L';f) \coloneqq D(1_\L;f) + D(1_{\L'};f) - D(1_{\L\cup\L'};f)\,.
\ee
Comparing the definitions of the Wiener--Hopf operator $W(1_\L) = 1_{\R^+}(X) 1_\L(-\mathrm{i}\dd/\dd x) 1_{\R^+}(X)$ and the localized ground-state projection $P(\L) = 1_\L 1_{\R^+}(-\mathrm{i}\dd/\dd x) 1_\L$, we see that the role of space and Fourier-space variables is exhanged and the order of the projections is permuted. Therefore, $W(1_\L)$ and $P(\L)$ have the same non-zero eigenvalues $\l_i$ including multiplicities. We would then guess that the operator $D(\L,\L';f)$ is trace class if and only if $\D(\L,\L';f)$ is trace class with the same trace. We do not prove this here since we are foremost interested in the properties of $\D(\L,\L';f)$ for certain H\"older-continuous functions $f$ on $[0,1]$ and of $D(a;f)$ for smooth $a$ and $f$. We write $f$ in the form $f = f_{0,\d} + (f-f_{0,\d})$ with a certain function $f_{0,\d}\in\mathsf{C}^2_c(\R)$ so that the trace of $\D(\L,\L';f_{0,\d})$ is close to the trace of $\D(\L,\L';f)$, see \eqref{continuity}. The symbol $1_\L$ on the other hand is replaced by a smooth symbol $\varphi_\eps\in\mathsf{C}^\infty(\R)$ which tends to $1_\L$ pointwise. Then we will use Widom's formula for the computation of the trace of $D(\varphi_\eps^2;f)$ and perform the limit $\eps\to0$ to obtain the trace of $\D(\L,\L';f)$.

For $\b\in\N$ and an open subset $\mathcal O\subseteq\R$, we denote by $\mathsf{C}^\b(\mathcal O)$ the vector space of $\b$-times differentiable, complex-valued functions on $\mathcal O$, whose derivative of the order $\b$ is continuous on $\widebar{\mathcal O}$. By $\mathsf{C}^0(\mathcal O)\coloneqq\mathsf{C}(\mathcal O)$, we understand the vector space of continuous, complex-valued functions on $\widebar{O}$ and by $\mathsf{C}^\infty(\mathcal O)\subset\mathsf{C}(\mathcal O)$ we mean the subspace of arbitrarily often differentiable functions on $\mathcal O$. We say that a function $f$ on $\mathcal O$ is $\mathsf{C}^\b$-smooth if $f\in\mathsf{C}^\b(\mathcal O)$. If a function $f \in\mathsf{C}^\b(\mathcal O)$, for some $\b\in\N_0\cup\{\infty\}$, is compactly supported with support inside $\mathcal O$, then we indicate this by adding the lower index $c$ to these function spaces and write $f\in\mathsf{C}^\b_c(\mathcal O)$.

We further assume the subsets $\L$ and $\L'$ to be intervals and use the letters $I_1$ and $I_2$ for these sets instead, or the letters $\mathcal I_1$ and $\mathcal I_2$ for a finite union of intervals. 

Let us introduce the norm $\| f\|_{y,\g}$ for $f \in \mathsf{C}(\R)$, where $y\in\R$ and $\gamma \in [0,1]$. Then, if $f\in \mathsf{C}^2(\R\setminus \{y\})$, we set
\begin{align}\label{def norm}
\| f\|_{y,\gamma} \coloneqq  \max_{0 \le k \le 2} \sup  \big\{|f^{(k)}(x)| \lvert x-y\rvert^{-\gamma+k} : x \in \R \setminus\{ y\}\big\} \, . 
\end{align} 
Sobolev uses the notation $\1 f\1_2$ for this norm, see \cite{Sobolev17}. As we will vary the parameters $y$ and $\g$ throughout this proof, but not the parameter $n=2$ (maximum degree $k$ of differentiation), we choose the notation $\|f \|_{y,\gamma}$.

Here are our assumptions on the intervals and on the test function $f$. 
\begin{assumptions} \label{assump}
\begin{enumerate}
\item The intervals $I_1\coloneqq (a_1,b_1)$ and $I_2\coloneqq(a_2,b_2)$ have disjoint closures and $I_1$ is bounded. %In the following, we will always assume that $a_1<b_1<a_2<b_2$.
\item  $f\in \mathsf{C}_{c}(\R)$ and satisfies $f(0)=0$.
\item There is a finite set $\mathcal X = \{x_1,x_2,\ldots, x_n\}\subset [0,1]$ of points with $x_1=0,x_n=1$ and a $\gamma \in (0,1]$ such that $f$ can be written as a sum $f=\sum_{i=1}^nf_i$ with $f_i\in \mathsf{C}_c\big((x_i-2,x_i+2)\big)$ so that $\| f_i \|_{x_i,\g} < \infty$. 
\end{enumerate}
\end{assumptions}

\begin{remarks}\label{remarks Renyi}
\begin{enumerate}
\item The result is symmetric in the intervals $I_1,I_2$. We only choose to assume that $I_1$ is bounded to simplify one estimate.
\item The assumption that the functions $f$ and $f_i$ are compactly supported is actually not relevant, as only the values of $f$ (or $f_i$) on the interval $[0,1]$ affect the operator $\D(\L,\L';f)$ defined in \eqref{Delta} since $0\le P(\L)\le \mathds{1}$.
 
\item The condition $\| f\|_{y,\g}< \infty$ with $\g>0$ implies $f(y)=0$. Thus, we have $f_i(x_i)=0$. However, as $f_i(x_j)$ for $i\not=j$ is not required to vanish, we can achieve $f(x_i)\not= 0$ by choosing the functions $f_i$ appropriately. In particular, any $\mathsf{C}^2$-function $f$, which is supported inside $[-1,2]$ also satisfies this assumption for $\mathcal X=\{0,1\}$ and $\gamma=1$, as we shall construct now. Let $\zeta \in \mathsf{C}^\infty_c\big((-1/2,1/2)\big)$, $0\le\zeta\le 1$ with $\zeta=1$ on $(-1/4,1/4)$. We choose $f_1(t) \coloneqq f(t)-\zeta(t)f(0)$ and $f_2(t)\coloneqq \zeta(t) f(0)$, which ensures $f_1(0)=f_2(1)=0$.

\item The R\'enyi entropy function $h_\a$ satisfies the last conditions. To show that, we set $\mathcal X \coloneqq \{0,1\}$ and for all $\a <1$ we choose $\g=\a$, for $\a=1$ we may take any $\g<1$, and for $\alpha>1$, we choose $\gamma=1$. As $h_\a(0)=h_\a(1)=0$, we will just use a smooth partition of unity to construct $f_1$ and $f_2$. For $t \in \R$, we set $f_1(t) \coloneqq h_\a(t) \zeta(t)$ and $f_2(t) \coloneqq h_\a(t) (1-\zeta(t))$. 
\end{enumerate}
\end{remarks}

Our main result is 
\begin{thm}\label{main thm}
Suppose that the intervals $I_1 = (a_1,b_1), I_2=(a_2,b_2)$ and the function $f$ satisfy \autoref{assump}. Then, the operator $\D(I_1,I_2;f)$ is trace class and
\begin{align} \label{formula}
\tr \D(I_1,I_2;f) &=\frac{U(0,1;f)}{2\pi^2} \ln\Big[\frac{(a_2-a_1)(b_2-b_1)}{(a_2-b_1)(b_2-a_1)}\Big]\,.
\end{align}
\end{thm}

This is proved in~\autoref{section 3}. We end this introduction with a few 
\begin{remarks} \label{remarks 1.4}
\begin{enumerate}

\item The quantity $\tr \D(I_1,I_2;f)$ has a number of well-known symmetries as can be seen by the explicit answer. The so-called cross--ratio term, $\frac{|a_2-a_1||b_2-b_1|}{|a_2-b_1||b_2-a_1|}$, of the two intervals $I_1$ and $I_2$ is strictly larger than $1$ and invariant under translations, multiplications, and inversion. To be more precise, we define the translations $T_a: x\mapsto x+a$, multiplication $M_t: x\mapsto t x$, $t\not=0$, and inversion $(\cdot)^{-1}: x\mapsto 1/x$, $x>0$.\footnote{It is natural to map $\R$ to the unit circle by the Cayley transform as done in \cite{LX18, Xu21} and thus intervals on $\R$ to intervals on the unit circle.} Then,
\[ \tr \D(V(I_1),V(I_2);f) = \tr \D(I_1,I_2;f)\,,\quad V = T_a, M_t,(\cdot)^{-1}\,.
\]
A translation $T_a$ is (obviously) implemented on $\mathsf{L}^2(\R)$ by the unitary transformation $f\mapsto f(x+a)$, the multiplication $M_t$ by $f\mapsto \sqrt{|t|}f(xt)$, and the inversion on $\mathsf{L}^2(\R^+)$ by the unitary and self-adjoint transformation $W:(Wf)(x)\coloneqq x^{-1} f(x^{-1})$. Here we have the properties, $W^* = W$, $W1_{\L^{-1}}W = 1_\L$ for $\L\subset\R^+$, where $\L^{-1}:=\{1/x:x\in\L\}$, and $WPW = P$. 

\item The method of proof in~\cite{LX18}, developed for the von-Neumann case $\a=1$, does not easily generalize to all $\a>0$; in fact, we only know how to extend this method to $\a\in(0,1)$, let alone to the more general functions we allow in this theorem. See our discussion in \autoref{section: master thesis}.

\item Formula \eqref{formula} makes sense (as a limit) if one of the intervals is unbounded, that is, either $a_1=-\infty$ or $b_2=+\infty$. We point out that also the proof requires only one of these intervals to be bounded. The formula can be easily extended to several intervals. Then, in the case of the von-Neumann entropy with $\a=1$ this is the formula of Casini and Huerta \cite{CH05} proved in \cite{LX18}. In the appendix we shall prove the following generalization of~\autoref{main thm}: 
\end{enumerate}
\end{remarks}

\begin{thm}\label{theorem:many_intervals}
Let $I_k:=(a_k,b_k)$, $1\leq k\leq N$, be $N$ open, pairwise disjoint intervals.
Let $n_1,n_2\geq 1$ such that $n_1+n_2=N$ and let $\mathcal{P}_1,\mathcal{P}_2\subset\{1,\dots,N\}$ be two index
sets of respectively $n_1,n_2$ elements such that $\mathcal{P}_1\cup\mathcal{P}_2=\{1,\dots,N\}$ and
$\mathcal{P}_1\cap\mathcal{P}_2=\emptyset$.
We define $\calI_j:=\bigcup_{k\in\mathcal{P}_j}I_j$ for $j=1,2$, and assume $\mathcal I_1$ to be bounded. Then, under the same assumptions on the function $f$ as in~\autoref{main thm}, the operator $\D(\calI_1,\calI_2;f)$ is trace class and
\begin{align}\label{formula_generalized}
\tr \Delta(\calI_1,\calI_2;f)=
\frac{U(0,1;f)}{2\pi^2}\sum_{k\in\mathcal{P}_1}\sum_{\ell\in\mathcal{P}_2} \ln\Big[\frac{|a_k-a_\ell||b_k-b_\ell|}{|a_k-b_\ell||b_k-a_\ell|}\Big]\,.
\end{align}
\end{thm}

Our main application of \autoref{theorem:many_intervals} is when the test function $f$ is the R\'enyi entropy function $h_\a$ with any $\a>0$. Let us first define the $\a$-R\'enyi entanglement entropy of the ground state of massless Dirac fermions with respect to disjoint subsets $\L$ and $\L'$ of the real line,
\be \mathsf{EE}(\L,\L';\a) \coloneqq \tr\D(\L,\L';h_\a)\,,
\ee
whenever $\D(\L,\L';h_\a)$ is trace class. This definition includes, of course, the von-Neumann entanglement entropy since $\mathsf{EE}(\L,\L';1) = \mathsf{EE}(\L,\L')$, which we introduced below \eqref{Delta}. The result is formulated in 
\begin{cl} \label{corollar}
Under the same assumptions on the sets $\mathcal I_1$ and $\mathcal I_2$ as in \autoref{theorem:many_intervals}, the $\a$-R\'enyi entanglement entropy of the ground state of massless Dirac fermions with respect to $\mathcal I_1$ and $\mathcal I_2$ is well-defined and
\be \mathsf{EE}(\mathcal I_1,\mathcal I_2;\a) = \frac{1+\a}{12\a}\,\sum_{k\in\mathcal{P}_1}\sum_{\ell\in\mathcal{P}_2} \ln\Big[\frac{|a_k-a_\ell||b_k-b_\ell|}{|a_k-b_\ell||b_k-a_\ell|}\Big]\,.
\ee
\end{cl}

\begin{proof} The R\'enyi entropy function $h_\a$ satisfies the conditions of \autoref{assump} as we noticed in \autoref{remarks Renyi}. Hence, $\D(\mathcal I_1,\mathcal I_2;h_\a)$ is trace class and the quantity $\mathsf{EE}(\mathcal I_1,\mathcal I_2;\a)$ is well-defined. Moreover, the coefficient $U(0,1;h_\a)$ can be computed explicitly, namely $U(0,1;h_\a) = \pi^2 (1+\a)/(6\a)$, see \cite[Appendix in arXiv version]{LSS14} and \cite{Blanco2022}. The rest follows immediately by \autoref{theorem:many_intervals}.
\end{proof}

\section{Trace-class properties}

We say that a compact operator $A$ on a (separable) Hilbert space $\mathcal H$ is in the Schatten--von Neumann class $\mathcal S_p$ for some $p>0$ if its singular values, $s_n(A), n\in\N$, form a sequence in $\ell^p(\R)$. The singular values $s_n(A)$ of $A$ are the eigenvalues of $\sqrt{A^*A}$. We arrange them in decreasing order, $s_n(A)\ge s_{n+1}(A)$ for all $n\in\N$. We equip $\mathcal S_p$ with the (quasi-)norm
\begin{align} \|A\|_{\mathcal S_p} \coloneqq \Big(\sum_{n=1}^\infty s_n(A)^p\Big)^{1/p}\,,\quad A\in\mathcal S_p\,.
\end{align}
For $p\ge1$, the space $\mathcal S_p$ is then a normed space, while for $p\in(0,1)$ it is a so-called quasi-normed space since this quasi-norm does not satisfy the usual triangle inequality but instead the so-called $p$-triangle inequality, $\|A_1+A_2\|_{\mathcal S_p}^p \le \|A_1\|_{\mathcal S_p}^p + \|A_2\|_{\mathcal S_p}^p$. The vector space $\mathcal S_p$ forms an ideal in the sense that if $S_1$ and $S_2$ are bounded operators on $\mathcal H$ and $A\in\mathcal S_p$ then $S_1AS_2\in \mathcal S_p$ with (the H\"older inequality) $\| S_1AS_2\|_{\mathcal S_p} \le \|S_1\| \|S_2\| \|A\|_{\mathcal S_p}$, where $\|\cdot\|$ denotes the operator norm. We frequently use that the adjoint operator $A^*$ of $A$ has the same $p$-quasi norm as $A$, that is, $\|A^*\|_{\mathcal S_p} = \|A\|_{\mathcal S_p}$ for all $p>0$. The space $\mathcal S_1$ is the ideal of trace-class operators and a Banach space with norm $\|\cdot\|_{\mathcal S_1}$. If an operator $A\in\mathcal S_1$, then we say that $A$ is trace class. The ideal of compact operators may be written as $\mathcal S_\infty$. As a general reference to the above we recommend \cite{BS77} and \cite{B88}. 

For fixed operator $A$, the $p$ (quasi-)norm $\|A\|_{\mathcal S_p}$ is increasing with decreasing $p$. So it is harder to prove an upper bound on $\|A\|_{\mathcal S_p}$ if we lower $p$.

Besides the projection $P=1_{\R^+}(-\mathrm{i}\dd/\dd x)$, it is convenient from a notational point of view to introduce the projections $Q_1\coloneqq 1_{I_1}(X)$, $Q_2\coloneqq 1_{I_2}(X)$ and $Q\coloneqq Q_1+Q_2 = 1_{I}(X)$ with $I = I_1\cup I_2$. Then, assuming $f(0)=0$,
\begin{align} \D(I_1,I_2;f) &= f\big(P(I_1)\big) + f\big(P(I_2)\big) - f\big(P(I)\big) = f(Q_1PQ_1) + f(Q_2PQ_2) - f(QPQ) \nonumber
\\
&= f(Q_1PQ_1 + Q_2PQ_2) - f(QPQ)
\,,
\end{align}
since $I_1\cap I_2=\emptyset$, or equivalently $Q_1Q_2=0$. This is a crucial identity. If we took the right-hand side, that is, $f(Q_1PQ_1 + Q_2PQ_2) - f(QPQ)$, as our definition of $\D(I_1,I_2;f)$ then we would not need to assume $f(0)=0$.

\subsection{Polynomial test functions}\label{subsection polynomial}

In the simplest (non-trivial) case we consider the linear function $f(t) = t$ in $\Delta(I_1,I_2;f)$. Then,
\[ \Delta(I_1,I_2;t) = -T -T^*
\]
with the operator $T\coloneqq Q_1PQ_2: \mathsf{L}^2(I_2)\to \mathsf{L}^2(I_1)$ and its adjoint operator $T^*$, both extended trivially (by $0$) to act on $\mathsf{L}^2(\R)$. According to \eqref{kernel P}, $T$ has integral kernel
\be \label{kernel T} 
T(x,y) = -\frac{1}{2\pi \mathrm{i}} \,\frac{1}{x-y}\,,\quad x\in I_1\,,y\in I_2 
\ee
since the intervals $I_1$ and $I_2$ have a positive distance. For the same reason, this kernel is not singular (in fact, arbitrarily often differentiable) and the Hilbert--Schmidt norm of $T$ is bounded, that is,
\[ \big\|Q_1PQ_2\big\|_{\mathcal S_2}^2 = \frac{1}{(2\pi)^2} \int_{I_1} \dd x \int_{I_2} \dd y \,\frac{1}{(x-y)^2} <\infty\,.
\]
The operator $T$ is not only in the Hilbert--Schmidt class $\mathcal S_2$ but in any Schatten--von Neumann class $\mathcal S_p$ for any $p>0$. Since we use this statement a number of times we formulate this in the following

\begin{lemma} \label{T in Sp} The operator $T$ with integral kernel defined in \eqref{kernel T} and with the intervals $I_1$ and $I_2$ satisfying \autoref{assump} is in any Schatten--von Neumann class $\mathcal S_p$ for any $p>0$. More generally, let $f_1$ and $f_2$ be bounded, measurable functions with supports in the closures $\widebar{I_1}$ and $\widebar{I_2}$, respectively. Then, the operator $f_1(X)Pf_2(X)$ is in $\mathcal S_p$ with $\big\|f_1(X)Pf_2(X)\big\|_{\mathcal S_p} \le \|f_1\|_{\mathsf{L}^\infty(I_1)} \, \|f_2\|_{\mathsf{L}^\infty(I_2)} \|T\|_{\mathcal S_p}$ for any $p>0$. See \eqref{def f(X)} for the definition of $f_j(X)$.
\end{lemma}

\begin{proof} Let us first get rid of $f_1$ and $f_2$. We use $f_j(X) = f_j(X)Q_j = Q_jf_j(X)$, the H\"older inequality and the fact that the operator norm $\|f_j(X)\|$ equals the $\mathsf{L}^\infty$-norm $\|f_j\|_{\mathsf{L}^\infty(I_j)}$ for $j\in\{1,2\}$,
\begin{align} \big\|f_1(X)Pf_2(X)\big\|_{\mathcal S_p} & = \big\|f_1(X) Q_1 P Q_2f_2(X)\big\|_{\mathcal S_p}\le \|f_1(X)\| \, \big\|Q_1 PQ_2\big\|_{\mathcal S_p} \,\|f_2(X)\|\nonumber
\\
&=\|f_1\|_{\mathsf{L}^\infty(I_1)} \, \|f_2\|_{\mathsf{L}^\infty(I_2)} \big\|Q_1 PQ_2\big\|_{\mathcal S_p} \,,
\end{align}
which is the claimed inequality. We still have to show that $T\in\mathcal S_p$. To this end, we may slightly enlarge $I_1$ to $\widetilde{I_1}$ so that $\widetilde{I_1}$ is bounded and $\widetilde{I_1}$ still has a positive distance to ${I_2}$. Let $\varphi \in \mathsf{C}^\infty_c(\widetilde{I_1})$ be a cutoff function such that $\varphi(x)=1$ for $x\in I_1$. Then, $Q_1=Q_1\varphi(X)$ and $T = Q_1 \varphi(X) PQ_2$. The (standard) trick is to insert the operator $\mathds{1}_1 = (\mathds{1}_1+A)^{-\b}(\mathds{1}_1+A)^{\b}$ between $Q_1$ and $\varphi(X)P$, where $\mathds{1}_1$ is the identity operator on $\mathsf{L}^2(\widetilde{I_1})$ and $A$ is a differential operator on $\widetilde{I_1}$ whose eigenvalues are known explicitly and the integer $\b>0$ is properly chosen. A natural candidate for this is the Laplacian $A\coloneqq -\dd^2/\dd x^2$ on $\widetilde{I_1}$, say with Neumann boundary conditions. The inverse $(\mathds{1}_1+A)^{-1}$ is well-defined since $\mathds{1}_1+A$ is a strictly positive operator and can be inverted on $\mathsf{L}^2(\widetilde{I_1})$. Thus, the (integer) power $(\mathds{1}_1+A)^{-\b} = \big((\mathds{1}_1+A)^{-1}\big)^\b$ is well-defined on $\mathsf{L}^2(\widetilde{I_1})$. This operator may be extended by 0 to define an operator on $\mathsf{L}^2(\R)$. We avoid this by including the operator $\widetilde{Q_1}\coloneqq 1_{\widetilde{I_1}}(X)$ to the right of $(\mathds{1}_1+A)^{-\b}$. Then, we write
\[  Q_1  \varphi(X) P Q_2 = Q_1 (\mathds{1}_1+A)^{-\b} (\mathds{1}_1+A)^{\b} \varphi(X) P Q_2 = Q_1 (\mathds{1}_1+A)^{-\b} \widetilde{Q_1} \cdot \widetilde{Q_1} (\mathds{1}_1+A)^{\b} \varphi(X) P Q_2\,.
\] 
The eigenvalues (equal to the singular values) of $A$ are $(n\pi/|\widetilde{I_1}|)^2$, $n\in\N_0$. Then,
\be \big\|Q_1(\mathds{1}_1+A)^{-\b} \widetilde{Q_1}\big\|_{\mathcal S_p}^p \le 1 + \sum_{n=1}^\infty \left(1+\left(\frac{n \pi}{|\widetilde{I_1}|}\right)^2\right)^{-\b p} < 1+ (|\widetilde{I_1}|/\pi)^{2\b p} \sum_{n=1}^\infty n^{-2\b p} < \infty\,,
\ee
if $\b> 1/(2p)$. The operator $\widetilde{Q_1}(\mathds{1}_1+A)^\b  \varphi(X) P Q_2$ is bounded, in fact Hilbert--Schmidt. We remark that the operator $\varphi(X)$ is required in order to satisfy the boundary conditions of $A$. We estimate its Hilbert--Schmidt norm,
\begin{align} &\big\|\widetilde{Q_1} (\mathds{1}_1+A)^\b  \varphi(X) PQ_2\big\|_{\mathcal S_2} \nonumber\\
&\le \sum_{k=0}^\b  {\b \choose k} \big\|\widetilde{Q_1} A^{k}  \varphi(X)PQ_2\big\|_{\mathcal S_2}\nonumber
\\
&\le \sum_{k=0}^\b  {\b \choose k} \sum_{\ell=0}^{2k} {2k\choose\ell}\left \lVert \varphi^{(2k-\ell)} \right \rVert_{\mathsf{L}^\infty(\widetilde{I_1})} \left(\int_{\widetilde{I_1}} \dd x \int_{I_2} \dd y \,\left|\frac{\partial^{\ell}}{\partial x^{\ell}} \frac{1}{(2\pi)(x-y)}\right|^2\right)^{1/2}\nonumber
\\
&=\sum_{k=0}^\b  {\b \choose k}\sum_{\ell=0}^{2k} {2k\choose\ell} \left \lVert \varphi^{(2k-\ell)}\right \rVert_{\mathsf{L}^\infty(\widetilde{I_1})} \frac{\ell!}{2\pi} \left( \int_{\widetilde{I_1}} \dd x \int_{I_2} \dd y \,\left|\frac{1}{(x-y)^{1+\ell}}\right|^2\right)^{1/2}\eqqcolon N_{\b,\varphi}\,.
\end{align}
This $N_{\b,\varphi}$ is finite as long as $\widetilde{I_1}$ and $I_2$ are separated by a positive distance and $\widetilde{I_1}$ is bounded, which they are. Altogether, the operator $T$ is the product of an operator in $\mathcal S_p$ and a bounded operator and therefore in $\mathcal S_p$ with 
\begin{align} \label{T is in Sp eq}
\|T\|_{\mathcal S_p} &\le \big\|Q_1(\mathds{1}_1+A)^{-\b}\widetilde{Q_1}\big\|_{\mathcal S_p} \,\big\|\widetilde{Q_1}(\mathds{1}_1+A)^\b  \varphi(X) P Q_2\big\|_{\mathcal S_\infty}\nonumber
\\
&\le\big\|Q_1(\mathds{1}_1+A)^{-\b}\widetilde{Q_1}\big\|_{\mathcal S_p} \,\big\|\widetilde{Q_1}(\mathds{1}_1+A)^\b \varphi(X) P Q_2\big\|_{\mathcal S_2} \nonumber
\\
&\le \big\|Q_1(\mathds{1}_1+A)^{-\b}\widetilde{Q_1}\big\|_{\mathcal S_p} \, N_{\b,\varphi}<\infty\,,
\end{align}
if $\b>1/(2p)$. Here, we used the H\"older inequality and monotonicity of the norms.
\end{proof}

As a direct consequence of this lemma, the operator $\Delta(I_1,I_2;t)\in\mathcal S_p$ for any $p>0$ and 
\[ \|\Delta(I_1,I_2;t)\|_{\mathcal S_p}\le 2^{1/p}\|T\|_{\mathcal S_p}
\] 
for $p\le1$. A fancy version of our down-to-earth bounds are the bounds by Birman--Solomyak, see \cite{BS77}.

Let us treat one more example explicitly, namely the quadratic polynomial $f(t)=t^2$. Then, 
\begin{align*} \D(I_1,I_2;t^2)&= Q_1PQ_1PQ_1+Q_2PQ_2PQ_2 - (Q_1+Q_2)P(Q_1+Q_2)P(Q_1+Q_2)
\\
&=-\big(T^*T + TT^* + Q_1PT + T^*PQ_1 + TPQ_2+Q_2PT^*\big)\,.
\end{align*}
Again, by the ideal properties, all operators in the sum are in $\mathcal S_p$ and so is $\D(I_1,I_2;t^2)$ with the estimate $\|\D(I_1,I_2;t^2)\|_{\mathcal S_1}\le 6 \|T\|_{\mathcal S_1}$ and a similar bound for general $p\le1$. 

This generalizes to arbitrary polynomials $f(t)=t^m$, $m\in\N, m\ge3$. If we expand the powers in $\Delta(I_1,I_2;f)$ then it is a sum of $2^{m+1}-2$, terms which all include at least one factor of $T$ or $T^*$. As a result we obtain for general monomials,
\be \|\D(I_1,I_2;t^m)\|_{\mathcal S_1}\le (2^{m+1}-2) \|T\|_{\mathcal S_1}\, ,
\ee 
which extends to arbitrary polynomials $f$ (with $f(0)=0$) by the triangle inequality. This is usually not a good bound and certainly not for our next step when we consider functions $f$ satisfying \autoref{assump}. Since the integral kernel is smooth we can calculate the trace of $\D(I_1,I_2;f)$ by the integral over the diagonal of its integral kernel. 

The same steps can be performed to prove that $D(I_1,I_2;f)$ is trace class for any monomial $f$. Clearly, $D(I_1,I_2;t) = 0$ and
\[ W\big(f(1_{I_1}(-\mathrm{i} \,\mathrm d / \mathrm d x ))\big) + W\big(f(1_{I_2}(-\mathrm{i} \,\mathrm d / \mathrm d x ))\big) - W\big(f(1_{I_1\cup I_2}(-\mathrm{i} \,\mathrm d / \mathrm d x ))\big) = 0\,,
\]
for any (bounded, measurable) function $f$ with $f(0)=0$. Therefore,
\begin{align*}
D(I_1,I_2; t^m )= \big( 1_{\R^+} 1_{I_1} (-\mathrm{i} \,\mathrm d / \mathrm d x ) 1_{\R^+} \big)^m + \big( 1_{\R^+} 1_{I_2} (-\mathrm{i} \,\mathrm d / \mathrm d x ) 1_{\R^+} \big)^m - \big( 1_{\R^+} 1_{I_1 \cup I_2} (-\mathrm{i} \,\mathrm d / \mathrm d x ) 1_{\R^+} \big)^m\,.
\end{align*}
As above we introduce the operator $\tilde{T}\coloneqq 1_{I_1} (-\mathrm{i} \,\mathrm d / \mathrm d x ) 1_{\R^+}(X) 1_{I_2} (-\mathrm{i} \,\mathrm d / \mathrm d x )  :\mathsf{L}^2(I_2) \to \mathsf{L}^2(I_1)$, which acts in Fourier space. It has the integral kernel
\[
\tilde{T}(\xi,\eta) = -\frac{1}{2\pi \mathrm{i}} \,\frac{1}{\xi-\eta}\,,\quad \xi\in I_1\,,\eta\in I_2 
\]
and is in $\mathcal S_p$ for any $p>0$. Expanding the powers of $D(I_1,I_2; t^m)$ we always have $\tilde{T}$ or its adjoint $\tilde{T}^*$ as a factor and hence $D(I_1,I_2; t^m)$ is in $\mathcal S_p$ as well.

\subsection{General test functions satisfying \autoref{assump}}

In the general case we use a theorem by Alexander V.~Sobolev \cite[Theorem 2.4]{Sobolev17}. It applies in his notation to $n=2, q\le 1$, $\s\in(0,1]$, $(2-\s)^{-1}<q$, $\s<\g$ and self-adjoint operators $A,B$ such that $A-B\in \mathcal S_{\s q}$. Then, for $f\in\mathsf{C}^2(\R\setminus\{x_0\})\cap \mathsf{C}(\R)$ with support in $[x_0-R,x_0+R]$, $x_0\in\R$, $R>0$,   
\be \label{inequ Sobolev} \|f(A) - f(B)\|_{\mathcal S_q}\le C R^{\g-\s} \|f\|_{x_0,\g}\,\||A-B|^\s\|_{\mathcal S_q}\,,
\ee
where the constant $C$ is independent of the operators $A,B$, of the function $f$ and of the parameter $R$. We recall the norm $\| f\|_{x_0,\g}$ from \eqref{def norm},
\[ \| f\|_{x_0,\g} = \max_{0\le k\le 2} \sup \big\{|f^{(k)}(x)||x-x_0|^{-\g+k}: x\in \R\setminus\{x_0\}\big\}\,. 
\]
In what follows, we adhere to the convention of representing any constant present in the inequalities by the symbol $C$, even if its value changes from line to line. At this point, we want to remark that this norm is submultiplicative in the following way: Let $g_1,g_2  \in \mathsf{C}^2(\R\setminus\{x_0\})\cap \mathsf{C}(\R)$ for some $x_0\in\R$. Then, we have
\be \label{weird norm is sub multiplicative}
\| g_1g_2 \|_{x_0,\g} \le C \| g_1 \|_{x_0,\g} \| g_2 \|_{x_0,0} \, ,
\ee
as can be seen with the product rule on $gh$ and the reordering of factors,
\[
\big|g_1^{(k_1)}(x) g_2^{(k_2)}(x)\big|\, | x - x_0 |^{-\gamma + k_1+k_2} = \big|g_1^{(k_1)}(x)\big|\, | x -x_0 | ^{-\gamma + k_1} \times \big|g_2^{(k_2)}(x)\big|\, | x- x_0| ^{-0+k_2} 
\]  
for $0 \le k_1+k_2 \le k\in\{0,1,2\}$.

To extend this to the slightly more general functions in \autoref{assump} with $|\mathcal X|>1$, we use the sum representation $f= \sum_{i=1}^n f_i$, where $\|f_i \|_{x_i,\g} < \infty$ and the $q$-triangle inequality to deduce
\begin{align*} 
\big\| f(A) - f(B) \big\|_{\mathcal S_q}^q &\le  \sum_{i=1}^n \big\| f_i(A)-f_i(B) \big\|_{\mathcal S_q}^q \le  C \sum_{i=1}^n \big\| f_i \big\|_{x_i,\g}^q \,\big\||A-B|^\s\big\|_{\mathcal S_q}^q \,.
\end{align*}
Applying convexity of $x\mapsto x^{1/q}$ on $\R^+$ (since $q\le1$) we get
\begin{align*} 
\|f(A)-f(B)\|_{\mathcal S_q} &\le C n^{1/q-1} \Big( \sum_{i=1}^n \| f_i\|_{x_i,\g} \Big) \||A-B|^\s\|_{\mathcal S_q}\, .
\end{align*}
We choose $A\coloneqq Q_1PQ_1+Q_2PQ_2$ and $B\coloneqq QPQ$ with $Q=Q_1+Q_2$ (as above). We do not need this here but with this inequality we can even show that $\D(I_1,I_2;h_\a)\in\mathcal S_q$ for any $q>1/2$ by choosing $0<\s<2-1/q$ and $\s<\a$. 

We set $q \coloneqq1$. Let us begin by simplifying the right-hand side of the last inequality, specifically, we note $A-B= -T-T^*$ and thus
\[
\| | A-B |^\sigma \|_{\mathcal S_1} = \| A-B \|_{\mathcal S_\sigma}^\sigma = \| -T-T^* \|_{\mathcal S_\sigma}^\sigma \le 2 \| T \|_{\mathcal S_\sigma}^\sigma  < \infty  \,.
\]
Thus, we have shown
\be
\| \D(I_1,I_2;f) \|_{\mathcal S_1} = \| f(A)-f(B) \|_{\mathcal S_1} \le  C \Big( \sum_{i=1}^n \| f_i\|_{x_i,\g} \Big) \| T \|_{\mathcal S_\sigma}^\sigma < \infty\, , \label{q inequality}
\ee
which means that $\D(I_1,I_2;f)$ is trace class.

To compute the trace of $\D(I_1,I_2;f)$ for general $f$, we will split $f$ into a sum of a $\mathsf{C}^2$-function, which we will take care of in the following section and a function, which yields a small trace norm.   Consider a smooth cutoff function $\zeta \in \mathsf{C}^\infty_{\mathrm{c}} ( (-1,1))$ with $\zeta(x)=1$ for $x \in [-1/2,1/2]$ and for $i=1, \dots , n, \delta \in (0,1)$, and $x \in \R$, define
\[
\zeta_{i,\delta}(x) \coloneqq \zeta \big( (x-x_i)/\delta\big) \,.
\]
Thus, $\zeta_{i,\delta} \in \mathsf{C}^{\infty}_{\mathrm{c} } ((x_i-\delta,x_i+\delta))$, $\zeta_{i,\delta}(t)=1$ for $x_i -\delta/2 \le x \le x_i + \delta/2$ and $\|\zeta_{i,\delta}^{(k)} \|_\infty \le C \delta^{-k}$ for $k=0,1,2$. In particular, $\| \zeta_{i, \d} \|_{x_i,0} \le C$. For $x \in \R$ and $i=1,2, \dots, n$, we define $f_{i,\delta}(x) \coloneqq f_i(x) \zeta_{i,\delta}(x)$ and $f_{0,\delta}(x) \coloneqq f(x) - \sum_{i=1}^n f_{i,\delta}(x)$. We point out that $f_{0,\d}\in\mathsf{C}^2_c(\R)$ as the functions  $x \mapsto f_i(x) \big(1- \zeta_{i,\d}(x)\big)$ are $\mathsf{C}^2$-smooth. Due to the submultiplicativity \eqref{weird norm is sub multiplicative}, we conclude
\[
\| f_{i,\delta} \|_{x_i, \g} \le C \| f_i \|_{x_i ,\g }\, \|\zeta_{i,\d} \|_{x_i,0} \le C \| f_i \|_{x_i,\g} \, .
\]
Thus, as $f_{i,\d}$ has its support inside $[x_i-\d,x_i+\d]$, we can conclude with \eqref{inequ Sobolev}
\begin{align*} 
\|\D(I_1,I_2; f -f_{0,\d})\|_{\mathcal S_1}& \le \sum_{i=1}^n \| \D (I_1,I_2; f_{i,\d} ) \|_{\mathcal S_1} \\
& \le \sum_{i=1}^n C \d^{\g-\s} \| f_{i,\d}\|_{x_i,\g}  \,\|T\|_{\mathcal S_\s}^\s \\
& \le C \d^{\g-\s} \sum_{i=1}^n   \| f_{i}\|_{x_i,\g} \, \|T\|_{\mathcal S_\s}^\s \, < \infty\,. 
\end{align*}
Because $\g>\s$ we get that 
\be \label{continuity} \tr \D(I_1,I_2;f) = \lim_{\d\to0} \tr \D(I_1,I_2;f_{0,\delta})\,.
\ee
Since we know the trace with smooth test functions such as $f_{0,\d}$, we can perform the limit and obtain our main result.

\section{Computation of the trace for sufficiently smooth test functions}\label{section 3}

\subsection{Fourier representation}
The actual computation of the trace of $\D(I_1,I_2;f)$ is done by relating it to the Widom formula for the trace of a ``smooth" version of $D(I_1,I_2;f)$. We use the standard Fourier representation
\begin{align}\label{Fourier rep}
f(A)-f(B)&=\frac{\mathrm{i}}{2\pi} \int_\R \dd t\, t\hat{f}(t) \int_0^1\dd s\, \mathrm{e}^{\mathrm{i}st A}(A-B)\,\mathrm{e}^{\mathrm{i}t(1-s)B}
\end{align}
for bounded operators $A,B$ and functions $f$ such that the $\mathsf{L}^1$-norm of $t\hat{f}$,
\[ \|t\hat{f}\|_{\mathsf{L}^1(\R)} \coloneqq \int_\R \dd t \,|t\hat{f}(t)|
\]
is finite. Here is a quick proof of this representation. We write
\[f(x) -f(y) = \frac{1}{2\pi}\int_\R\dd t \,\hat{f}(t)\Big(\mathrm{e}^{\mathrm{i} tx}  - \mathrm{e}^{\mathrm{i} ty} \Big)\,,\quad x,y\in\R\,,
\]
and use Duhamel's formula,
\[ \mathrm{e}^{\mathrm{i} tA} -  \mathrm{e}^{\mathrm{i} tB} = \int_0^1\dd s\, \frac{\dd}{\dd s} \mathrm{e}^{\mathrm{i} t (sA+(1-s)B)} =  \mathrm{i}t\int_0^1\dd s\,\Big(\mathrm{e}^{\mathrm{i} tsA} (A-B) \,\mathrm{e}^{\mathrm{i} t(1-s) B} \Big)\,.
\]
Altogether, under the condition $\|t\hat{f}\|_{\mathsf{L}^1(\R)}<\infty$, we proved the representation~\eqref{Fourier rep}.

If, in addition $A-B$ is trace class, then $f(A)-f(B)$ is also trace class and the estimate 
\[ \|f(A)-f(B)\|_{\mathcal S_1}\le \frac{1}{2\pi}\|t\hat{f}\|_{\mathsf{L}^1(\R)} \|A-B\|_{\mathcal S_1}
\]
follows. 

Before we continue we make the following 
\begin{remark} Suppose that $f \in \mathsf{C}^{\b+1}_{\mathrm{c}}(\R)$ for $\b\in\{1,2\}$. Then, by the Cauchy--Schwarz inequality we get
\begin{align} \label{Sobolev norm}
\| t^\b \hat {f} \|_{\mathsf{L}^1(\R)} &\le \| (1+t^2)^{\b/2} \hat{f} \|_{\mathsf{L}^1(\R)} \le \| (1+t^2)^{(\b+1)/2} \hat{f} \|_{\mathsf{L}^2(\R)}\, \| 1/\sqrt{1+t^2} \|_{\mathsf{L}^2(\R)} \nonumber
\\&\le \| (1+t^2)^{(\b+1)/2} \hat{f} \|_{\mathsf{L}^2(\R)} \,C \le  C \| f \|_{\mathsf{H}^{\b+1}(\R)} \,.
\end{align}
%\begin{align} \label{Sobolev norm}
%\| t^\b \hat {f} \|_{\mathsf{L}^1(\R)} &\le \| t^{\b-1}\sqrt{1+t^2} \hat{f} \|_{\mathsf{L}^1(\R)} \le \| t^{\b-1}(1+t^2) \hat{f} \|_{\mathsf{L}^2(\R)}\, \| 1/\sqrt{1+t^2} \|_{\mathsf{L}^2(\R)} \nonumber
%\\&\le \| t^{\b-1}(1+t^2) \hat{f} \|_{\mathsf{L}^2(\R)} \,C \le C \| f \|_{\mathsf{H}^{\b+1}(\R)} \,,
%\end{align}
The last Sobolev norm, $\| f \|_{\mathsf{H}^{\b+1}(\R)} $, is finite, which implies $\| t^\b \hat{f} \|_{\mathsf{L}^1(\R)} < \infty$.
\end{remark}

With this Fourier representation, we can show (again) that for $f\in\mathsf{C}^{2}_c(\R)$, the operator $\D(I_1,I_2;f)$ is trace class if we use $A = Q_1PQ_1+Q_2PQ_2$ and $B= QPQ$ as above. This is a side remark but this representation will be used again shortly. 

\subsection{Smooth version and the relation to Widom's formula}\label{section:smooth_Widom}

Let us introduce smooth versions $Q_{1,\eps}$ and $Q_{2,\eps}$ of the projections $Q_1$ and $Q_2$, respectively of the corresponding indicator functions. To this end, for any $0 < \eps <\operatorname{dist}(I_1,I_2)/4$, we define the intervals $I_{j,\eps} \coloneqq (a_j+\eps,b_j-\eps) \subset I_j$ for $j=1,2$. Let $\varphi_{1,\eps}\in\mathsf{C}^\infty(\R)$ and $\varphi_{2,\eps}\in \mathsf{C}^\infty(\R)$ be cutoff functions such that $\varphi_{j,\eps} \equiv 0$ on $I_{j}^\complement$, $\varphi_{j,\eps}\equiv 1$ on $I_{j,\eps}$, and $0 \le \varphi_{j,\eps}\le 1$ for $j=1,2$. Thus, $\varphi_{j,\eps}$ converges pointwise to $1_{I_j}$ as $\eps \to 0$. We may write $Q_{1,\eps} = \varphi_{1,\eps}(X)$ and $Q_{2,\eps} = \varphi_{2,\eps}(X)$ to denote the multiplication operators by the functions $\varphi_{1,\eps}$ and $\varphi_{2,\eps}$. We also set $Q_\eps\coloneqq Q_{1,\eps}+Q_{2,\eps} = \varphi_\eps(X)$ with the function $\varphi_\eps\coloneqq \varphi_{1,\eps}+\varphi_{2,\eps}$.

Let
\begin{align} A_\eps\coloneqq Q_{1,\eps} P Q_{1,\eps} + Q_{2,\eps} P Q_{2,\eps},\quad B_\eps\coloneqq Q_{\eps} P Q_{\eps}\,,
\end{align}
and $A = Q_1PQ_1 + Q_2PQ_2$ and $B= QPQ$ as above. 

The operator $f(A)-f(B)$ is trace-class since $A-B = -Q_1PQ_2-Q_2PQ_1 = -T -T^* \in \mathcal S_1$ by \autoref{T in Sp}. The same is true for the operator $f(A_\eps)-f(B_\eps)$ since $A_\eps-B_\eps = -\varphi_{1,\eps}(X)P\varphi_{2,\eps}(X) - \varphi_{2,\eps}(X)P\varphi_{1,\eps}(X)\in\mathcal S_1$ by the general version of \autoref{T in Sp} with $f_1=\varphi_{1,\eps}$ and $f_2=\varphi_{2,\eps}$. 

Using \eqref{Fourier rep}, we write
\begin{align}
[f(A_\eps)-&f(B_\eps) - f(A)+f(B)]\\
&= \int_\R \dd t\, t \hat{f}(t) \int_0^1\dd s\, \left(\mathrm{e}^{\mathrm{i} st A_\eps}(A_\eps-B_\eps) \,\mathrm{e}^{\mathrm{i}(1-s)tB_\eps} - \mathrm{e}^{\mathrm{i} st A}(A-B) \,\mathrm{e}^{\mathrm{i}(1-s)tB}\right). \label{3.5}
\end{align}
We want to show that this integral converges to $0$ in trace norm as $\eps \to 0$. As $\lVert t \hat{f}(t) \rVert_{\mathsf{L}^1(\R)} < \infty$, by dominated convergence, it suffices to show that, in trace norm, the integrand is uniformly bounded and converges pointwise to $0$. To this end, we will use the following 

\begin{lemma} \label{strong conv and trace norm lem}
 Let $S_n,Z_n$ be two sequences of uniformly (in operator norm) bounded operators on a Hilbert space $(\mathcal H,\|\cdot\|_{\mathcal H})$ and let $Z_n$ converge to $0$ strongly. Furthermore, let $\mathcal R$ be a trace class operator on $\mathcal H$. Then, the operator sequences $Z_n \mathcal R S_n$ and $S_n \mathcal R Z_n $ are uniformly bounded in $\mathcal S_1$ and converge to $0$ in $\mathcal S_1$. 
\end{lemma}

For the sake of completeness, we provide a quick proof of this well-known fact. See \cite[Theorem A.1]{GPS} for a slightly different proof. 
\begin{proof}
As $(Z_n \mathcal R S_n )^*= S_n^* \mathcal R^* Z_n^*$, it suffices to show that the claim holds for $Z_n \mathcal R S_n$. We observe (recall, $\|Z\|$ stands for the operator norm of $Z$)
\[ \| Z_n \mathcal R S_n \|_{\mathcal S_1} \le \| Z_n \mathcal R \|_{\mathcal S_1} \,\sup_{n \in  \N} \| S_n \| \le \sup_{n\in \N} \| Z_n \| \,\| \mathcal R \|_{\mathcal S_1} \sup_{n\in \N} \,\| S_n \| \,. \]
Thus, the sequence is uniformly bounded in $\mathcal S_1$ and it suffices to show that $Z_n \mathcal R \to 0$ in $\mathcal S_1$. For this, let $(\psi_m)_{m \in \N}$ be an orthonormal basis of the orthogonal complement of the kernel of $\mathcal R$, which recovers the singular values of $\mathcal R$, meaning $\| \mathcal R \psi_m \|_{\mathcal H} = s_m(\mathcal R)$. Let $P_m$ be the projection onto $\psi_m$. We shall now use the triangle inequality to obtain 
\[
 \| Z_n \mathcal R \|_{\mathcal S_1} = \Big\| Z_n \mathcal R  \sum_{m \in \N} P_m  \Big\|_{\mathcal S_1} \le \sum_{m \in \N} \| Z_n \mathcal R P_m \|_{\mathcal S_1} = \sum_{m \in \N} \| Z_n \mathcal R \psi_m \|_{\mathcal H} \, . 
\]
As $n \to \infty$, the $m^{\rm{th}}$ entry of this sum converges to $0$, as $Z_n$ converges strongly to $0$. As the $m^{\rm{th}}$ entry is always bounded by $\sup_{n \in \N} \| Z_n \| \,s_m(\mathcal R)$, the sum converges to $0$ by dominated convergence.
\end{proof}

Let us get back to the integrand in \eqref{3.5}. Let $\tau\coloneqq st$ and $\tau'\coloneqq (1-s)t$. Then, 
\begin{align} \mathrm{e}^{\mathrm{i} \tau A_\eps}(A_\eps-B_\eps) \,\mathrm{e}^{\mathrm{i}\tau'B_\eps} - \mathrm{e}^{\mathrm{i} \tau A}(A-B) \,\mathrm{e}^{\mathrm{i}\tau'B} &=\mathrm{e}^{\mathrm{i} \tau A_\eps}(A_\eps-B_\eps-A+B) \,\mathrm{e}^{\mathrm{i}\tau'B_\eps}\label{1st term}
\\
&+(\mathrm{e}^{\mathrm{i} \tau A_\eps} - \mathrm{e}^{\mathrm{i} \tau A})(A-B)\,\mathrm{e}^{\mathrm{i}\tau'B_\eps}\label{2nd term}
\\
&+\mathrm{e}^{\mathrm{i} \tau A}(A-B)(\mathrm{e}^{i \tau' B_\eps} - \mathrm{e}^{\mathrm{i} \tau' B})\,.\label{3rd term}
\end{align}
We shall first see that the second term \eqref{2nd term} is of the form mentioned above. The operator $A-B$ is trace class. We note that $A_\eps$ converges to $A$ strongly as $\eps\to0$. This is because as a multiplication operator, $(\varphi_{\eps}-1_I)\psi$ converges in $\mathsf{L}^2(\R)$ to 0 for any $\psi\in\mathsf{L}^2(\R)$. Therefore, also $\mathrm{e}^{\mathrm{i}\tau A_\eps}$ tends to $\mathrm{e}^{\mathrm{i}\tau A}$ strongly. Moreover, the sequence $\mathrm{e}^{\mathrm{i}\tau A_\eps} - \mathrm{e}^{\mathrm{i}\tau A}$ is uniformly bounded by 2. Finally, $\mathrm e^{\mathrm i \tau' B_\eps}$ is a unitary operator. Thus, we can apply \autoref{strong conv and trace norm lem} for this summand. For the third term \eqref{3rd term} the same arguments apply.

Finally, for the trace norm of the first term on the right-hand side of \eqref{1st term} we estimate

\begin{align*}
\big\|\mathrm{e}^{\mathrm{i} \tau A_\eps}&(A_\eps-B_\eps-A+B) \mathrm{e}^{\mathrm{i}\tau'B_\eps}\big\|_{\mathcal S_1} \le \big\|A_\eps-B_\eps-A+B\big\|_{\mathcal S_1}
\\
&\le\big\|Q_{1,\eps} P Q_{2,\eps} - Q_{1} P Q_{2}\big\|_{\mathcal S_1} + \big\|Q_{2,\eps} P Q_{1,\eps} - Q_{2} P Q_{1}\big\|_{\mathcal S_1} \\
&=2\big\|Q_{1,\eps} P Q_{2,\eps} - Q_1PQ_2\big\|_{\mathcal S_1}\\
&\le 2 \big\|(Q_{1,\eps} -Q_1)Q_1P Q_2 Q_{2,\eps}\big\|_{\mathcal S_1} + 2 \big\|Q_{1} P Q_{2}(Q_{2,\eps}-Q_2)\big\|_{\mathcal S_1}\,.
\end{align*}
We used that $Q_{j,\eps}=Q_{j,\eps} Q_j$, as $\operatorname{supp}(\varphi_{j,\eps}) \subset I_j$. Here, the operator $\mathcal R = Q_1PQ_2=T$ is trace class by \autoref{T in Sp}. The uniform boundedness follows as $0 \le \varphi_{j, \eps} \le 1$ and we already established the strong convergence of $Q_{j,\eps}-Q_j$ to $0$. Thus, the terms inside the trace norm converge to $0$ in trace norm due to \autoref{strong conv and trace norm lem} and thus, we have shown that the integrand in \eqref{3.5} converges pointwise to $0$ in trace norm and is uniformly bounded in trace norm, which implies that the integral converges to $0$ using dominated convergence.

For the next step we use the following
\begin{lemma} Suppose that $\hat{f}\in\mathsf{L}^1(\R)\cap\mathsf{L}^2(\R)$, $t\mapsto t^2\hat{f}(t)\in \mathsf{L}^1(\R)$, and $f(0)=0$. Then, with the above definitions of $Q_{1,\eps}$, $Q_{2,\eps}$, $Q_{\eps} = Q_{1,\eps}+Q_{2,\eps}$, and $P$, we have
\begin{align} \tr &\big[f(Q_{1,\eps}PQ_{1,\eps}) + f(Q_{2,\eps} PQ_{2,\eps}) - f(Q_\eps PQ_\eps)\big] \label{E_n}
\\&= \tr\big[f(PQ_{1,\eps}^2P) + f(PQ_{2,\eps}^2P) - f(PQ_\eps^2 P)\big]\label{F_n}
\\
&= \label{3.17a}\tr\big[f(PQ_{1,\eps}^2P) - Pf(Q_{1,\eps}^2)P\big]  + \tr\big[f(PQ_{2,\eps}^2P) - Pf(Q_{2,\eps}^2)P)\big] 
\\
&- \tr \big[f(PQ_\eps^2 P) - Pf(Q_\eps^2) P)\big]\,.\label{3.17b}
\end{align}
\end{lemma}

\begin{remark} Below we will consider $f\in\mathsf{C}_c^3(\R)$. By \eqref{Sobolev norm} the above integrability conditions on $f$ are then satisfied.
\end{remark}

\begin{proof} We begin this proof with the polynomial test functions $f(t)= t^n$, $n\in\N$. Strictly speaking, no such function satisfies the integrability conditions stated in the Lemma. However, we may replace $f$ by $f\varphi$, where $\varphi\in\mathsf{C}_c^3(\R)$ with support in $[-1,2]$ and the property $\varphi(t) = 1$ for $t\in[0,1]$. The function $f\varphi$ satisfies all the asked integrability conditions. Since all the involved operators $A\in\{Q_{1,\eps}PQ_{1,\eps},\ldots,PQ_\eps^2 P\}$ have their spectra inside $[0,1]$, we have $f(A) = (f\varphi)(A)$ for these operators.

For each $n\in\N$, let us introduce the operators
\begin{align}\label{def E_n} 
E_n&\coloneqq (Q_{1,\eps}PQ_{1,\eps})^n + (Q_{2,\eps} PQ_{2,\eps})^n - (Q_\eps PQ_\eps)^n\,,
\\\label{def F_n} 
F_n&\coloneqq (PQ_{1,\eps}^2P)^n + (PQ_{2,\eps}^2 P)^n - (PQ_\eps^2 P)^n\,.
\end{align}
Then these are the operators inside the traces in \eqref{E_n} and \eqref{F_n}. We need to show that $E_n$ and $F_n$ are trace class and have the same trace. Note for a start that $E_1 = - Q_{1,\eps}PQ_{2,\eps} - Q_{2,\eps}PQ_{1,\eps}$ is trace class by \autoref{T in Sp} with $f_1 = \varphi_{1,\eps}$ and $f_2 = \varphi_{2,\eps}$ and vice versa. %The integral kernel of $E_1$ is smooth and therefore the trace of $E_1$ can be computed as the integral over the diagonal of this kernel, which is zero. 
The operator $F_1=0$ is trivially trace class and with the same (zero) trace as $E_1$ (see the arguments below \eqref{trace equality}).

For general $n\in\N$, we have the following recursion relations
\begin{align} PQ_\eps E_n Q_\eps P &= F_{n+1}\,,\label{recursion 1}
\\
Q_\eps P F_n P Q_\eps & = E_{n+1}\label{recursion 2} 
\\
&+ Q_{1,\eps}P(PQ_{1,\eps}^2P)^nPQ_{2,\eps} + Q_{1,\eps}P(PQ_{2,\eps}^2P)^nPQ_{2,\eps}\label{remainder 1} 
\\
&+Q_{2,\eps}P(PQ_{1,\eps}^2P)^nPQ_{1,\eps} + Q_{2,\eps}P(PQ_{2,\eps}^2P)^nPQ_{1,\eps}\label{remainder 2} 
\\
&+Q_{2,\eps}P(PQ_{1,\eps}^2P)^nPQ_{2,\eps} + Q_{1,\eps}P(PQ_{2,\eps}^2P)^nPQ_{1,\eps}\,.\label{remainder 3} 
\end{align}
To prove these relations, we merely use $Q_\eps Q_{j,\eps} = Q_{j,\eps}$ for $j\in\{1,2\}$. 

Let us first consider the remainder terms in (\ref{remainder 1}--\ref{remainder 3}). All of them contain the factor $T_{\eps}\coloneqq Q_{1,\eps}PQ_{2,\eps}= \varphi_{1,\eps}(X) P \varphi_{2,\eps}(X)$ or the adjoints $T_{\eps}^*$. These operators are trace class by \autoref{T in Sp}. The remaining factors are all bounded (in operator norm) by 1. Therefore, by the H\"older inequality,
\begin{align*} \big\|\eqref{remainder 1}\big\|_{\mathcal S_1}& \le \big\|Q_{1,\eps}P(PQ_{1,\eps}^2P)^{n-1}P Q_{1,\eps}T_{\eps}\big\|_{\mathcal S_1} + \big\|T_{\eps}Q_{2,\eps}P(PQ_{2,\eps}^2P)^{n-1}PQ_{2,\eps}\big\|_{\mathcal S_1}
\\
&\le 2\big\|T_{\eps}\big\|_{\mathcal S_1} \eqqcolon C/3\,.
\end{align*}
The same bound applies to the other two remainder terms so that
\be \label{remainder estimate} \big\|\eqref{remainder 1} + \eqref{remainder 2} + \eqref{remainder 3} \big\|_{\mathcal S_1} \le C\,,
\ee
independent of $n$. Together with the above recursion relations, this yields the bounds
\begin{align} \label{norm of En and Fn}
\|E_n\|_{\mathcal S_1} \le C n\,,\quad \|F_n\|_{\mathcal S_1} \le C n
\end{align}
for some constant $C$ as we shall show now. Equation \eqref{recursion 1} and the H\"older inequality gives the bound (i) $\|F_{n+1}\|_{\mathcal S_1} \le \|E_n\|_{\mathcal S_1}$. Equation \eqref{recursion 2} and the above estimate on the remainder terms show the bound (ii) $\|E_{n+1}\|_{\mathcal S_1} \le \|F_n\|_{\mathcal S_1} + C$ with the constant $C$ in \eqref{remainder estimate}. We combine and iterate the inequalities (i) and (ii) and obtain
\[ \|F_{n+1}\|_{\mathcal S_1} \le \|F_{n-1}\|_{\mathcal S_1} + C \le \|F_{n-3}\|_{\mathcal S_1} + 2C \le \ldots\,.
\]
We end up either at $F_1=0$ or at $F_2 = -PQ_{1,\eps}^2PQ_{2,\eps}^2P - PQ_{2,\eps}^2PQ_{1,\eps}^2P$. By \autoref{T in Sp},
\[ \big\|F_2\big\|_{\mathcal S_1} = 2 \big\|P Q_{1,\eps}Q_{1,\eps}PQ_{2,\eps}Q_{2,\eps}P\big\|_{\mathcal S_1} \le 2 \big\|Q_{1,\eps}PQ_{2,\eps}\big\|_{\mathcal S_1} = 2 \big\|T_{\eps}\big\|_{\mathcal S_1} \le C<\infty\,.
\]
This proves the 1-norm estimate on $F_n$ in \eqref{norm of En and Fn}. But since we also have the estimate $\|E_{n+1}\|_{\mathcal S_1} \le \|F_n\|_{\mathcal S_1} + C$ we have also proved the first bound in \eqref{norm of En and Fn} on the 1-norms of $E_n$.

Now we come to the computation of the trace of $E_n$ and of $F_n$. In the first step we use the cyclicity $\tr (AB) = \tr(BA)$, provided that $AB$ and $BA$ are both trace class. Here, $A=PQ_{\eps}$ and $B=E_nQ_\eps P$ are, of course, bounded. The operators $AB$ and $BA$ are indeed trace class because $E_n$ is trace class. Then, by \eqref{recursion 1}
\begin{align} \label{trace equality}&\tr F_{n+1}=\tr \big[PQ_{\eps}E_nQ_\eps P\big] \nonumber
\\
&=\tr\big[E_n Q_\eps PQ_\eps\big]\nonumber
\\
&=\tr\big[(Q_{1,\eps}PQ_{1,\eps})^n Q_{1,\eps}P(Q_{1,\eps}+Q_{2,\eps}) + (Q_{2,\eps}PQ_{2,\eps})^n Q_{2,\eps}P(Q_{1,\eps}+Q_{2,\eps}) - (Q_\eps PQ_\eps)^{n+1}\big]\nonumber
\\
&=\tr\big[E_{n+1} + (Q_{1,\eps}PQ_{1,\eps})^n Q_{1,\eps}PQ_{2,\eps} + (Q_{2,\eps}PQ_{2,\eps})^n Q_{2,\eps}PQ_{1,\eps}\big]\nonumber
\\
&=\tr E_{n+1}\,.
\end{align}
The last equality follows since the trace-class operators $(Q_{j,\eps}PQ_{j,\eps})^n Q_{j,\eps}PQ_{j',\eps}$ with $j\not=j'$ have zero trace. There are various ways to see the latter. One is to use that the smooth integral kernel vanishes on the diagonal and apply Mercer's theorem, see \cite{B88}. Another one uses an orthonormal basis $(\psi_n)_{n\in\N}$ with support either in $I_1$, in $I_2$, or in $(I_1\cup I_2)^\complement$ so that any $\langle \psi_n,(Q_{j,\eps}PQ_{j,\eps})^n Q_{j,\eps}PQ_{j',\eps}\psi_n\rangle = 0$. A third method is to use cyclicity and permute $Q_{j,\eps}$ inside the trace to hit $Q_{j',\eps}$ so that $Q_{j,\eps}Q_{j',\eps} = 0$ appears as a factor. 

In the general case with $\hat{f}\in\mathsf{L}^1(\R)\cap\mathsf{L}^2(\R)$, we use the Fourier representation
\[ f(A) = \frac{1}{2\pi} \int_\R \dd x \,\hat{f}(x) \,\mathrm{e}^{\mathrm{i} x A}
\]
for the bounded operators $A$ equal to $Q_{1,\eps}PQ_{2,\eps}$, $Q_{2,\eps}PQ_{2,\eps}$ and $Q_{\eps}PQ_{\eps}$ and expand the exponentials. That leads to 

\begin{align}\tr&\big[f(Q_{1,\eps}PQ_{1,\eps}) + f(Q_{2,\eps}PQ_{2,\eps}) - f(Q_{\eps}PQ_{\eps})\big]
\\
&=\frac{1}{2\pi} \int_\R \dd x \,\hat{f}(x) \,\tr\big[ \mathrm{e}^{\mathrm{i} x Q_{1,\eps}PQ_{1,\eps}} +  \mathrm{e}^{\mathrm{i} x Q_{2,\eps}PQ_{2,\eps}} -  \mathrm{e}^{\mathrm{i} x Q_{\eps}PQ_{\eps}}- \mathds{1} \big] \label{3.20}
\\
&=\frac{1}{2\pi} \int_\R \dd x \,\hat{f}(x) \sum_{n\ge1} \frac{(\mathrm{i}x)^n}{n!}\tr\big[(Q_{1,\eps}PQ_{1,\eps})^n + (Q_{2,\eps}PQ_{2,\eps})^n - (Q_{\eps}PQ_{\eps})^n\big]\label{3.21}
\\
&=\frac{1}{2\pi} \int_\R \dd x \,\hat{f}(x) \sum_{n\ge1} \frac{(\mathrm{i}x)^n}{n!}\tr\big[(PQ_{1,\eps}^2P)^n + (PQ_{2,\eps}^2P)^n - (PQ_\eps^2P)^n\big] \label{3.22}
\\
&=\frac{1}{2\pi} \int_\R \dd x \,\hat{f}(x) \sum_{n\ge1} \frac{(\mathrm{i}x)^n}{n!}\tr\big[P(PQ_{1,\eps}^2P)^n + P(PQ_{2,\eps}^2P)^n - P(PQ_\eps^2P)^n\big]
\\
&=\frac{1}{2\pi} \int_\R \dd x \,\hat{f}(x)\,\tr\big[P\mathrm{e}^{\mathrm{i} x PQ_{1,\eps}^2P} +  P\mathrm{e}^{\mathrm{i} x PQ_{2,\eps}^2P} -  P\mathrm{e}^{\mathrm{i} x PQ_\eps^2P}- P \big] \label{3.23}
\\
&=\tr\big[Pf(PQ_{1,\eps}^2P) + Pf(PQ_{2,\eps}^2P)-Pf(PQ_\eps^2P)\big]\label{3.23b}
\\
&=\tr\big[f(PQ_{1,\eps}^2P) + f(PQ_{2,\eps}^2P)-f(PQ_\eps^2P)\big]\label{3.24}\,.
\end{align}
In the last step we have used that for any orthogonal projection $P$ and any self-adjoint, bounded operator $A$ on the Hilbert space $\mathcal H$, we have for $u\in\ker(P)$,
\begin{align*} (\mathds{1}-P) f(PAP) u = (\mathds{1}-P)f(0)u = f(0)u\,.
\end{align*}
This is zero as $f(0)=0$. On the other hand, if $Pu=u$, then $f(PAP)u\in P\mathcal H$ and $(\mathds{1}-P)f(PAP)u=0$. Any $u\in\mathcal H$ can be written as $u = (\mathds{1}-P)u + Pu$ with $(\mathds{1}-P)u \in\ker(P)$. Hence, $f(PAP) = Pf(PAP) + (\mathds{1}-P)f(PAP) = Pf(PAP)$.

In the following, we present the arguments that justify all the above interchanges of traces, integrals and sums. In the very first step and in \eqref{3.23b} we used that $\frac{1}{2\pi}\int_\R \dd x \hat{f}(x) = f(0) = 0$ so that we could smuggle in the unity operator $\mathds{1}$ and get rid of $P$, respectively. Then, for the interchange of the trace and the integral we used $t\mapsto t^2\hat{f}(t) \in\mathsf{L}^1(\R)$ and that $A(x)\coloneqq \mathrm{e}^{\mathrm{i} x Q_{1,\eps}PQ_{1,\eps}} +  \mathrm{e}^{\mathrm{i} x Q_{2,\eps}PQ_{2,\eps}} -  \mathrm{e}^{\mathrm{i} x Q_{\eps}PQ_{\eps}} - \mathds{1}$ is trace class. To see the latter, we write
\begin{align} A(x)&=\exp\big(\mathrm{i} x(Q_{1,\eps}PQ_{1,\eps}+Q_{2,\eps}PQ_{2,\eps})\big) - \exp\big(\mathrm{i} x Q_{\eps}PQ_{\eps}\big)\nonumber
\\
&= e_x(Q_{1,\eps}PQ_{1,\eps}+Q_{2,\eps}PQ_{2,\eps}) - e_x(Q_{\eps}PQ_{\eps})
\end{align}
with a function $e_x\in\mathsf{C}^2_c(\R)$ defined as
\begin{align} e_x(t)\coloneqq \left\{\begin{array}{ll} \exp(\mathrm{i} x t) &\mbox{ for } t\in [0,1]\\0&\mbox{ for } t\in \R\setminus [-1,2]\end{array}\right.\,.
\end{align} 
To be more precise, let $\varphi\in\mathsf{C}^3(\R)$ be a function of compact support equal to  $[-1,2]$ and equal to 1 on $[0,1]$. Then, $e_x(t) = \varphi(t) \exp(\mathrm{i} xt)$ and the $\|\cdot\|_{0,0}$-norm of $e_x$ defined in \eqref{def norm} is finite with 
\[ \|e_x\|_{0,0} \le C (1+ |x| + |x|^2)\le C(1+|x|^2)\,.
\]
An application of inequality \eqref{inequ Sobolev} tells us that with some constant $C$,
\begin{align} \label{A inequality}
\|A(x)\|_{\mathcal S_1} &\le C \|e_x\|_{0,0} \big\|Q_{1,\eps}PQ_{2,\eps} + Q_{2,\eps}PQ_{1,\eps}\big\|_{\mathcal S_1} \le C (1 + |x|^2) \big\|Q_{1,\eps}PQ_{2,\eps}\big\|_{\mathcal S_1}\nonumber
\\
&\le C (1 + |x|^2)
\end{align}
by \autoref{T in Sp} in the last step. %By the inequality $\|Q_{1,\eps} P Q_{2,\eps} \|_{\mathcal S_1} \le \|Q_1 P Q_2 \|_{\mathcal S_1}$ that we used in \eqref{Teps in Sp} and by \autoref{T in Sp}, the operator $Q_{1,\eps}PQ_{2,\eps}$ on the right-hand side of \eqref{A inequality} is trace class.} 
Altogether,
\[ \int_\R \dd x \, |\hat{f}(x)| \,\|A(x)\|_{\mathcal S_1} \le C \int_\R \dd  x \, |\hat{f}(x)| \,(1+|x|^2) <  \infty
\]
by our assumption on $f$. Thus we have justified \eqref{3.20}. 

In order to arrive at the next equality, \eqref{3.21}, we interchange the Taylor series of the exponentials for fixed $x$, which are well-defined, with the trace. It is clear that the series starts at $n=1$. Let 
\begin{align} B(x) &\coloneqq \sum_{n\ge1} \frac{(\mathrm{i} x)^n}{n!} \big[(Q_{1,\eps}PQ_{1,\eps})^n + (Q_{2,\eps}PQ_{2,\eps})^n - (Q_{\eps}PQ_{\eps})^n\big]=\sum_{n\ge1} \frac{(\mathrm{i} x)^n}{n!} E_n\,.
\end{align}
We use the bound \eqref{norm of En and Fn} to show that
\begin{align} \|B(x)\|_{\mathcal S_1} &\le \sum_{n\ge1} \frac{|x|^n}{n!} \, \big\|E_n\big\|_{\mathcal S_1}\le C \sum_{n\ge1} \frac{|x|^n}{n!} n = C |x|\exp(|x|)<\infty\,.
\end{align}
This bound is sufficient as $x$ is fixed and thus justifies \eqref{3.21}. 

In order to get to \eqref{3.22} we use \eqref{norm of En and Fn}. Then,
\begin{align} 
|\hat{f}(x)| \, \Big\|\sum_{n\ge1} \frac{(\mathrm{i} x)^n}{n!} F_{n}\Big\|_{\mathcal S_1} &\le |\hat{f}(x)| \, \sum_{n\ge1} \frac{|x|^n}{n!} \big\|F_{n}\big\|_{\mathcal S_1} \le C |\hat{f}(x)| \,\sum_{n\ge1} \frac{|x|^n}{n!} n = C  |\hat{f}(x)| \,|x|\exp(|x|)\,.
\end{align}
Again, this bound is sufficient for fixed $x$. For the last step to go from \eqref{3.23} to \eqref{3.23b} and bring the trace in front of the integral we define for $x\in\R$ the operator
\begin{align}G(x)\coloneqq P\mathrm{e}^{\mathrm{i} x PQ_{1,\eps}^2P} +  P\mathrm{e}^{\mathrm{i} x PQ_{2,\eps}^2P} -  P\mathrm{e}^{\mathrm{i} x PQ_\eps^2P}- P\,.
\end{align}
We go into Fourier space and write
\begin{align} \mathcal F^{-1} \, G(x) \,\mathcal F &= 1_{\R^+}\exp\big(\mathrm{i} x W(\varphi_{1,\eps}^2)\big) +  1_{\R^+}\exp\big(\mathrm{i} x W(\varphi_{2,\eps}^2)\big) -  1_{\R^+}\exp\big(\mathrm{i} x W(\varphi_{\eps}^2)\big) - 1_{\R^+} %\st{\mathds{1}}
\end{align}
with the Wiener--Hopf operators $W(\cdot)$ as defined in \eqref{WH}. We claim that $G(x)$, or equivalently $\mathcal F^{-1} G(x)\,\mathcal F$, is trace class. To see this, we use \eqref{exp W} for each exponential with the Hankel operators $H(\cdot)$ as defined in \eqref{Hankel}. Then,
\begin{align} \mathcal F^{-1}\,G(x)\,\mathcal F &= W\big(\exp(\mathrm{i}x \varphi_{1,\eps}^2)\big) + W\big(\exp(\mathrm{i}x \varphi_{2,\eps}^2)\big) -W\big(\exp(\mathrm{i}x \varphi_\eps^2)\big) - 1_{\R^+}\label{3.43}
\\
&-\mathrm{i}\int_0^x \dd y\, \Big[H\Big(\exp\big(\mathrm{i}y\varphi_{1,\eps}^2\big),\varphi_{1,\eps}^2\Big)\,\exp\big(\mathrm{i}(x-y)W(\varphi_{1,\eps}^2)\big) 
\\
&\phantom{\mathrm{i}\int_0^x \dd y\,\Big]} + H\Big(\exp\big(\mathrm{i}y\varphi_{2,\eps}^2\big),\varphi_{2,\eps}^2\Big)\,\exp\big(\mathrm{i}(x-y)W(\varphi_{2,\eps}^2)\big) 
\\
&\phantom{\mathrm{i}\int_0^x \dd y\,\Big]}- H\Big(\exp\big(\mathrm{i}y\varphi_\eps^2\big),\varphi_\eps^2\Big)\,\exp\big(\mathrm{i}(x-y)W(\varphi_\eps^2)\big)\Big] \,.
\end{align}
The whole expression on the right-hand side of \eqref{3.43} vanishes because it is a Wiener--Hopf operator with the symbol $\exp\big(\mathrm{i}x\varphi_{1,\eps}^2\big) + \exp\big(\mathrm{i}x\varphi_{2,\eps}^2\big) - \exp\big(\mathrm{i}x\varphi_{\eps}^2\big) -1$ identically equal to 0. The trace norm of the remaining terms are estimated as follows:
\begin{align} \Big\|\int_0^x \dd y\, &\Big[H\Big(\exp\big(\mathrm{i}y\varphi_{1,\eps}^2\big),\varphi_{1,\eps}^2\Big)\,\exp\big(\mathrm{i}(x-y)W(\varphi_{1,\eps}^2)\big)\Big\|_{\mathcal S_1}
\\
&\le \int_0^{|x|} \dd y\, \Big\|H\Big(\exp\big(\mathrm{i}y\varphi_{1,\eps}^2\big),\varphi_{1,\eps}^2\Big)\Big\|_{\mathcal S_1} \, \Big\|\exp\big(\mathrm{i}(x-y)W(\varphi_{1,\eps}^2)\big)\Big\|
\\
&\le \int_0^{|x|} \dd y\, \interleave \exp\big(\mathrm{i}y\varphi_{1,\eps}^2\big)\interleave \, \interleave \varphi_{1,\eps}^2\interleave
\\
&\le \int_0^{|x|} \dd y\, |y|\,\interleave \varphi_{1,\eps}^2\interleave^2 =\frac12 |x|^2 \, \interleave \varphi_{1,\eps}^2\interleave^2\,.
\end{align}
Here, we have used \eqref{1.13}. The Besov-norm of $\varphi_{1,\eps}^2$ is finite since $\eps>0$. The same arguments apply to the two other Hankel operators involved. Altogether, $\|G(x)\|_{\mathcal S_1} \le C |x|^2$, where $C$ depends on $\eps$, and
\[ \int_\R \dd x\, |\hat{f}(x)| \, \|G(x)\|_{\mathcal S_1} \le C  \int_\R \dd x\, |\hat{f}(x)| \, |x|^2 <\infty
\]
by our assumption on $f$. This allows for interchanging the trace and the integral and we have justified all steps. 

To prove the last claim of this lemma we insert $Pf(Q_{1,\eps}^2)P + Pf(Q_{2,\eps}^2)P - Pf(Q_\eps^2)P = 0$ and then, 
\begin{align*} f(PQ_{1,\eps}^2P) &- Pf(Q_{1,\eps}^2)P +  f(PQ_{2,\eps}^2P) - Pf(Q_{2,\eps}^2)P -  f(PQ_{\eps}^2P) + Pf(Q_{\eps}^2)P\\
& = D(\varphi_{1,\eps}^2;f) + D(\varphi_{2,\eps}^2;f) - D(\varphi_{\eps}^2;f)\,.
\end{align*}
By Widom's result, stated below \eqref{def D(a;f)}, each single operator difference in \eqref{3.17a} and \eqref{3.17b} --- such as $f(PQ_{1,\eps}^2P) - Pf(Q_{1,\eps}^2)P$ --- is trace class. By linearity of the trace we can put the trace in front of each of these three terms. This finishes the proof.
\end{proof}

This is now in the form where we can apply Widom's formula. That is, 
\be \eqref{3.17a} + \eqref{3.17b} = \mathcal B(\varphi_{1,\eps}^2;f) + \mathcal B(\varphi_{2,\eps}^2;f) - \mathcal B(\varphi_\eps^2;f)\,.
\ee
It remains to compute its limit $\eps\to0$. This is content of 

\begin{lemma} \label{trace smooth} Suppose $f\in\mathsf{C}_c^{3}(\R)$ and $f(0)=0$. Then, under the above conditions on the functions $\varphi_{1,\eps}$ and $\varphi_{1,\eps}$ we have
\begin{align} \lim_{\eps\to0} \big[\mathcal B(\varphi_{1,\eps}^2;f) + \mathcal B(\varphi_{2,\eps}^2;f) - \mathcal B(\varphi_\eps^2;f)\big] = \frac{U(0,1;f)}{2\pi^2} \ln\Big[\frac{(b_2-b_1)(a_2-a_1)}{(a_2-b_1)(b_2-a_1)}\Big]\,.
\end{align}
\end{lemma}

\begin{proof} We recall the definition of $\mathcal B(a;f)$ from \eqref{Widom} and of the function $U$ from \eqref{def U}. To perform the integration with respect to $\xi_1$ and $\xi_2$, we split the real axis into the three regions $I_1$, $I_2$ and $\R\setminus (I_1\cup I_2)$. Note that $U(\s_1,\s_2;f) = 0$ if $\s_1=\s_2$. Hence, $U\big(\varphi_{j,\eps}^2(\xi_1),\varphi_{j,\eps}^2(\xi_2);f\big) = 0$ if both $\xi_1$ and $\xi_2$ are in $I_{j,\eps}$ or both are in $\R\setminus I_j$, for $j\in\{1,2\}$, because the first two entries of $U$ are either both $1$ or both $0$. Similarly, $U\big(\varphi_{\eps}^2(\xi_1),\varphi_\eps^2(\xi_2);f\big) = 0$ if both $\xi_1$ and $\xi_2$ are in $I_{1,\eps}\cup I_{2,\eps}$ or both are in $\R\setminus (I_1\cup I_2)$.

By the symmetry $U(\s_1,\s_2;f) = U(\s_2,\s_1;f)$ we can thus write 
\begin{align*} &\mathcal B(\varphi_{1,\eps}^2;f) + \mathcal B(\varphi_{2,\eps}^2;f)
\\
&=\frac{1}{4\pi^2}\int_{I_1}\dd \xi_1\int_{\R\setminus I_1}\dd \xi_2\, \frac{U\big(\varphi_{1,\eps}^2(\xi_1),\varphi_{1,\eps}^2(\xi_2);f\big)}{(\xi_1-\xi_2)^2} + \frac{1}{4\pi^2}\int_{I_2}\dd \xi_1\int_{\R\setminus I_2}\dd \xi_2\, \frac{U\big(\varphi_{2,\eps}^2(\xi_1),\varphi_{2,\eps}^2(\xi_2);f\big)}{(\xi_1-\xi_2)^2}
\end{align*}
and
\begin{align*} \mathcal B(\varphi_{\eps}^2;f)&=\frac{1}{4\pi^2}\int_{I_1\cup I_2}\dd \xi_1\int_{\R\setminus (I_1\cup I_2)}\dd \xi_2\, \frac{U\big(\varphi_{\eps}^2(\xi_1),\varphi_{\eps}^2(\xi_2);f\big)}{(\xi_1-\xi_2)^2}\,.
\end{align*}
In the difference $\big(\mathcal B(\varphi_{1,\eps}^2;f) + \mathcal B(\varphi_{2,\eps}^2;f)\big) - \mathcal B(\varphi_{\eps}^2;f)$ we are left with the integration over $(\xi_1\in I_1,\xi_2\in I_2)$ and $(\xi_1\in I_2,\xi_2\in I_1)$. There are no singular terms as $I_1$ and $I_2$ have a positive distance and the limit $\eps\to0$ poses no problem. Therefore,
\begin{align}&\lim_{\eps\to0} \big[\mathcal B(\varphi_{1,\eps}^2;f) + \mathcal B(\varphi_{2,\eps}^2;f) - \mathcal B(\varphi_\eps^2;f)\big]
\\
&=\frac{U(1,0;f)}{4\pi^2}\int_{a_1}^{b_1} \dd\xi_1\int_{a_2}^{b_2}\mathrm{d}\xi_2\, \frac{1}{(\xi_1-\xi_2)^2} + \frac{U(1,0;f)}{4\pi^2} \int_{a_2}^{b_2} \mathrm{d}\xi_1 \int_{a_1}^{b_1} \mathrm{d}\xi_2\, \frac{1}{(\xi_1-\xi_2)^2}
\\
&=\frac{U(0,1;f)}{2\pi^2} \int_{a_1}^{b_1} \mathrm{d}\xi_1 \left(-\frac{1}{b_2-\xi_1} + \frac{1}{a_2-\xi_1}\right)
\\
&=\frac{U(0,1;f)}{2\pi^2} \ln\Big[\frac{(b_2-b_1)(a_2-a_1)}{(a_2-b_1)(b_2-a_1)}\Big]\,.
\end{align}

\end{proof}

Finally, we present the
\begin{proof}[Proof of \autoref{main thm}] First of all, we use \eqref{q inequality} which shows that $\D(I_1,I_2;f)\in\mathcal S_1$. By continuity \eqref{continuity}, we only need to compute the trace with the ``smooth" test function $f(1-\zeta_\d)\in \mathsf{C}_c^3(\R)$, which we just accomplished in \autoref{trace smooth}; the function $\zeta_\d$ is localized around the points in $\mathcal X$.
\end{proof}

\section{Discussion of results}

\subsection{Integral representation of R\'enyi entropies \texorpdfstring{$h_\a$ for $\a\le1$}{h\_^^ce^^b1 for  ^^ce^^b1  ^^e2 ^^89 ^^a4   1}}
\label{section: master thesis}

As mentioned, the computation of Casini--Huerta and the proof of Longo--Xu rely on an integral representation of the von-Neumann entropy $h_1$ and on the operator-concavity\footnote{A real-valued, measurable function $f$ on $\R$ is called operator concave if for all bounded self-adjoint operators $A$ and $B$ on a Hilbert space $(\mathcal H,\langle\cdot,\cdot\rangle)$ and real numbers $t\in(0,1)$, $f(tA+ (1-t)B)\ge t f(A) + (1-t) f(B)$. Positivity of operators is meant as the positivity of the inner product, that is, $A\ge 0$ if $\langle \varphi, A\varphi\rangle\ge0$ for all $\varphi\in \mathcal H$.} property of the function $h_1$; see \cite[Lemma 3.4]{LSS23} for a proof of this property. This can be extended to the $\a$-\Renyi{} entropies for $\a\in(0,1)$. Here, we only provide a short sketch, while we refer to Ref.~\cite{Ferro22} for detailed proofs of our results. 

To develop a suitable integral representation of the \Renyi{} entropy function, we interpret $h_{\a}$ as an analytic function defined on the complex domain $D\coloneqq\Cp\cup\Cm\cup (0,1)$, taking the principal branch of the logarithm $\ln(z)$ and of the power function $z^{\a}$, where we denoted by $\Cpm:=\{z\in\C|\pm\mathrm{imag}({z)}>0\}$ the open upper and lower complex half-planes with the real line removed.

We recall (see \cite{Donoghue1974}) that any analytic self-map in the upper complex half-plane $m:\Cp\rightarrow\C$ with $m(\Cp)\subseteq\overline{\Pi}_+$ is a \NH{} function (NH in short) and it may be uniquely represented by the sum of a linear and an integral term. Rearranging the argument of the logarithm in the complex \Renyi{} entropy function, we rewrite it as
\begin{align}\label{eq:Htildealphalogsum}
h_{\a}(z)&=\frac{1}{1-\a}\ln\left(z^\a\left(1+\left(\frac{1-z}{z}\right)^\a\right)\right)
=\frac{\a}{1-\a}\ln({z})+\frac{1}{1-\a}L_{\a}(-z)\,,
\end{align}
where the function $L_{\a}:z\mapsto\ln\left(1+\left(\frac{1+z}{-z}\right)^{\a}\right)$
is defined on the same complex domain as $h_{\a}$. The function $z\mapsto\frac{1}{-z}$,
the logarithm $\ln({z})$ and, for $\a\in (0,1)$, the power function $z\mapsto z^{\a}$ are NH, 
and therefore the function $L_{\a}$ is NH for $\a\in (0,1)$. Combining the unique \NH{} representations of the logarithm and of $L_{\a}$,
it follows from \eqref{eq:Htildealphalogsum} that the complex R\'enyi function $h_{\a}$ on $D$ admits the integral representation:
\begin{align}\label{eq:Ha_integral_representation}
h_{\a}(z)=\frac{B({\a})}{1-\a}-\frac{1}{1-\a}\int^{+\infty}_{\frac{1}{2}}\mathrm{d}\lambda\,&
f_{\a}(\lambda)\Bigg(R_z(\lambda)-R_z(-\lambda)
+\frac{\frac{1}{2}-\lambda}{\left(\frac{1}{2}-\lambda\right)^2+1}
-\frac{\frac{1}{2}+\lambda}{\left(\frac{1}{2}+\lambda\right)^2+1}
\Bigg)\,,
\end{align}
where
\begin{align}\label{eq:beta_Renyi_term2_result}
B(\a)=\frac{1}{2}\ln\left(1+2^{\a}+2^{\frac{\a}{2}+1}\cos\Big(\frac{3}{4}\a\pi\Big)\right),\hspace{0.5cm}
f_{\a}(\lambda):=\frac{1}{\pi}\arctan\left(\frac{\left(\frac{2\lambda-1}{2\lambda+1}\right)^{\a}\sin({\a\pi})}
{1+\left(\frac{2\lambda-1}{2\lambda+1}\right)^{\a}\cos({\a\pi})}\right)\,,
\end{align}
and $R_z$ is the function
\begin{align}\label{eq:resolvent}
R_z:\R\setminus\left(-\frac{1}{2},\frac{1}{2}\right)\rightarrow\C\,,\hspace{0.5cm}
\lambda\mapsto R_z(\lambda):=\frac{1}{z-\frac{1}{2}+\lambda}\,.
\end{align}

\begin{remarks}
\begin{enumerate}
\item The whole information concerning the entropy function is summarized in the function
$\frac{f_{\a}}{1-\a}$ which may be interpreted as the derivative of the generating function
of a measure in the \NH{} integral representation. On the other hand, the
term $\frac{B(\a)}{1-\a}$ as well as the remaining integral terms
independent of $z$ are of little interest, since they are constant and therefore
they do not contribute to the R\'enyi entanglement entropy.

\item For R\'enyi indices $\a>1$, the method described above cannot be employed.
The difficulty arises from the fact that, in this case, the main branch of the complex
power function $z\mapsto z^{\a}$ is not NH.

\item In the von-Neumann limit $\a\uparrow1$, the integral representation in~\eqref{eq:Ha_integral_representation}
reduces to the formula provided in Ref.~\cite{CH05}, for $t\in(0,1)$,
\begin{align}\label{von Neumann}
h_1(t) =  -t\ln t-(1-t)\ln(1-t)&=-\int^{+\infty}_{\frac{1}{2}}\mathrm{d}\lambda\,
\left(\left(\lambda-\frac{1}{2}\right)\left(R_t(\lambda)-R_t(-\lambda)\right)
-\frac{2\lambda}{\lambda+\frac{1}{2}}\right)\,.
\end{align}
\end{enumerate}
\end{remarks}

A key aspect of the integral representation \eqref{von Neumann} lies in the difference $R_t(\lambda)-R_t(-\lambda)$ in the integrand. When examining the operator $h_1(P(\mathcal I))$, where $P(\mathcal I)$ was defined in~\eqref{restricted}, the two functions $R_t(\pm\lambda)$
give rise to the resolvent operators of $P(\mathcal I)$ at the points $\frac{1}{2}\mp\lambda$, which are known explicitly for $\mathcal I$ being a union of intervals as in \autoref{theorem:many_intervals}. This is crucial in the proof provided in Refs.~\cite{CH05,LX18}.
Our integral representation~\eqref{eq:Ha_integral_representation} also contains the same difference
$R_t(\lambda)-R_t(-\lambda)$ in the integral. The proof of Longo--Xu can be carried out replacing $h_1$ with $h_{\a}$ and hence we proved \autoref{theorem:many_intervals} for the R\'enyi entropy functions $h_\a$ with $\a\in(0,1)$, see \cite{Ferro22}.

\subsection{Small and large separation of intervals}
An interesting aspect of formula \eqref{formula} is the behavior when the two intervals get separated by a large distance or get close together. Let us elaborate this a bit. In the first case, let us assume that $I_1=(a_1,b_1)$ with $a_1<b_1$ and $I_2=(a_2+r,b_2+r)$ with $r>0$, $a_2<b_2<\infty$ and let us look at the asymptotics of the $\ln$-coefficient as $r\to\infty$. It turns out that as $r\to\infty$,
\begin{align*}  
\frac{(a_2+r-a_1)(b_2+r-b_1)}{(a_2+r-b_1)(b_2+r-a_1)}&=1+\frac{1}{r^2} (b_1-a_1)(b_2-a_2) + O(r^{-3})
\end{align*}
so that 
\[\ln\Big[\frac{(a_2+r-a_1)(b_2+r-b_1)}{(a_2+r-b_1)(b_2+r-a_1)}\Big] = \frac{1}{r^2} |I_1| |I_2| + O(r^{-3})\,.
\]
If one of the intervals is unbounded, say $I_2 = [a_2,\infty)$, then the $\ln$-coefficient equals $\ln\big[\frac{a_2-a_1}{a_2-b_1}\big]$. If we now separate the two intervals as above and replace $a_2$ by $a_2+r$, then
\[ \ln\Big[\frac{a_2+r-a_1}{a_2+r-b_1}\Big] = \frac{1}{r} |I_1| + O(r^{-2})\,.
\]
This tells us that in the case of large separating intervals $I_1$ and $I_2$ as above, the R\'enyi entanglement entropy behaves to leading order in the separating distance $r$ as $(1+\a)/(12\a) |I_1||I_2|/r^2$, if both intervals are bounded and as $(1+\a)/(12\a) |I_1|/r$ if $I_1$ is bounded and $I_2$ is unbounded.

\medskip
On the other hand, if the intervals $I_1$ and $I_2$ get close to one another, the entanglememt entropy diverges logarithmically. That is, suppose that $I_1=(a_1,b_1)$ and $I_2=(b_1+\eps,b_2)$ with $\eps>0$. Then 
\[\frac{(b_1+\eps-a_1)(b_2-b_1)}{\eps(b_2-a_1)} = O(1/\eps)\,.
\]
In the case that the closures of the intervals overlap at just a single point $b_1$ (put $\eps=0$ above), then the crucial operator $1_{I_1}(X) P 1_{I_2}$ is not trace class and our further estimates are out of reach. It can be seen immediately that this operator is not Hilbert--Schmidt. In fact, it will not be in any Schatten--von Neumann class $\mathcal S_p$ for any $p<\infty$.

\subsection{A generalization to intersecting domains}

In the earlier paper \cite{CH04} (without the explicit formula for the von Neumann entropy), Casini and Huerta introduced the trace of the following operator
\be \label{def F}
F(\L,\L';f)\coloneqq f(P(\L)) + f(P(\L')) - f(P(\L\cap \L')) - f(P(\L\cup \L'))
\ee
in the case of the von-Neumann entropy function $f=h_1$ and for not necessarily disjoint Borel sets $\L,\L'$. See also the recent review \cite{Xu21}. They showed various nice properties of $\tr F(\L,\L';h_1)$ as a function of the sets $\L, \L'$, in particular, its positivity as a consequence of strong subadditivity of quantum mechanical entropy. The latter property was (first) proved by Lieb and Ruskai in \cite{LR}. In general, the $\a$-R\'enyi entropy is not strongly subadditive and positivity of the $\a$-R\'enyi entanglement entropy is hence not clear. In any case, not much is known in general whether $F(\L,\L';f)$ is trace class.

Let us mention the following property. Suppose that $f$ satisfies $f(0)=0$ (as usual) and the symmetry $f(t)=f(1-t)$, like for the R\'enyi entropies. Then, the operator $f(P(\L))$ has the same eigenvalues as $f(P(\L^\complement))$ including multiplicities (except for the eigenvalue 0) and informally we have that $\tr F(\L,\L';f) = \tr F(\L^\complement,(\L')^\complement;f)$. %This also entails that $ F(\L,\L';f)$ is trace class if and only if $F(\L^\complement,(\L')^\complement;f)$ is trace class.

This implies (informally) that in the case of (disjoint) sets $\L$ and $\L'$ satisfying our \autoref{assump},
\begin{align*}  \tr \D(\L,\L';f) &= \tr F(\L,\L';f)  
\\
&= \tr\big[f(P(\L^\complement)) + f(P((\L')^\complement)) - f(P(\L^\complement\cap (\L')^\complement)) \big]
\\
&=\tr F(\L^\complement,(\L')^\complement;f) \,.
\end{align*}

There is a trivial but useful expression that relates the operator $F(\L,\L';f)$ to the difference of two $\D(\cdot,\cdot;f)$ operators with certain disjoint sets. More concretely, we have 

\begin{lemma} \label{F Delta} For any Borel subsets $\L$ and $\L'$ of $\R$ we have
\be F(\L,\L';f) = F(\L\setminus\L',\L';f) - F(\L\setminus\L',\L\cap\L';f) = \D(\L\setminus\L',\L';f) - \D(\L\setminus\L',\L\cap\L';f) \,.
\ee
Note that the sets in the arguments of $\D$ are disjoint, but not necessarily their closures.
\end{lemma}

We find it convenient to change our notation and write in the remaining of this section instead of $1_\L P 1_{\Lambda'}$ simply $\L P\L'$. 

\begin{proof} The first equality is (almost) tautological as 
\begin{align*} &F(\L\setminus\L',\L';f) - F(\L\setminus\L',\L\cap\L';f) 
\\
&= f\big((\L\setminus \L')P(\L\setminus \L')\big) + f(\L' P\L') - f\big(((\L\setminus\L')\cup \L')P((\L\setminus\L')\cup \L')\big) 
\\
&- f\big((\L\setminus \L')P(\L\setminus \L')\big) - f\big((\L\cap\L')P(\L\cap\L')\big) + f\big(((\L\setminus \L')\cup(\L\cap\L'))P((\L\setminus \L')\cup(\L\cap\L'))\big)
\\
&=f(\L' P\L') - f\big((\L\cup\L')P(\L\cup\L')\big) - f\big((\L\cap\L')P(\L\cap\L')\big) + f(\L P\L)
\\
&=F(\L,\L';f)\,.
\end{align*}
The second equality follows trivially from the definition.
\end{proof}

As in the main body of this paper, we assume that $\L$ and $\L'$ are intervals and we use the letters $I_1$ and $I_2$. 

\begin{conj} \label{conject} Suppose that the intersection of the closures of the (open) intervals $I_1$ and $I_2$ has positive (finite) Lebesgue measure. In addition, we assume that one of the non-empty sets $I_1\setminus I_2$ or $I_2\setminus I_1$ is bounded. Then, for $f$ satisfying the \autoref{assump}, the operator $F(I_1,I_2;f)$ defined in \eqref{def F} is trace class and
\be \label{trace of F}
\tr F(I_1,I_2;f) = \frac{U(0,1;f)}{2\pi^2} \ln \Big[\frac{|I_1||I_2|}{|I_1\cap I_2| |I_1\cup I_2|}\Big]\,.
\ee
\end{conj}

\begin{remarks}
\begin{enumerate}
\item 
This formula is in agreement with formula (17) in \cite{CH04} and with Theorem 5.1(1) in \cite{Xu21} for the von-Neumann entropy function $f=h_1$.
\item Under the given conditions, one of the intervals must be bounded. The right-hand side of \eqref{trace of F} is, of course, well-defined if both $I_1$ and $I_2$ are bounded. Suppose that $I_1$ is bounded and $I_2$ is unbounded, then the argument of the logarithm is understood as $\frac{|I_1|}{|I_1\cap I_2|}$. 
\end{enumerate}
\end{remarks}

\begin{proof}[Proof for polynomial test functions $f(t) = t^m$, $m\in\N$.] In the first part we show that $F(I_1,I_2;t^m)$ is trace class. Note that we cannot use \autoref{F Delta} here since $\D(I_1\setminus I_2,I_2;t^m)$ and $\D(I_1\setminus I_2,I_1\cap I_2;t^m)$ are not trace class. We will use that relation to compute the trace of $F(I_1,I_2;t^m)$ though. 

For a start, let us consider the linear test function $f(t)=t$. Then, by some simple calculations,
\begin{align} F(I_1,I_2;t) &= \D(I_1\setminus I_2,I_2;t) - \D(I_1\setminus I_2,I_1\cap I_2;t) \nonumber
\\
&= - (I_1\setminus I_2) P (I_2\setminus I_1) - (I_2\setminus I_1) P (I_1\setminus I_2) = -T-T^*\,,
\end{align}
where the operator $T\coloneqq (I_1\setminus I_2) P (I_2\setminus I_1)$ is defined as in \autoref{subsection polynomial} with $I_1$ replaced by $I_1\setminus I_2$ and $I_2$ replaced by $I_2\setminus I_1$. Note that there is a positive distance between the intervals $I_1\setminus I_2$ and $I_2\setminus I_1$ of $|I_1\cap I_2|$ and therefore $F(I_1,I_2;t)$ is in the Schatten class $\mathcal S_p$ for any $p>0$. 

Now we deal with the polynomial $f(t)=t^m$ by induction over $m$. We claim that $F(I_1,I_2;t^m)$ is trace class for every integer $m \ge 1$. We have just seen the starting case $m=1$. Thus, we consider
\begin{align}
&F(I_1,I_2;t^{m+1} ) \nonumber\\
&= (I_1PI_1)^m P (I_1 \cup I_2 - I_2\setminus I_1) +(I_2PI_2)^m P (I_1 \cup I_2 - I_1\setminus I_2) \nonumber \\
&-  [(I_1\cap I_2)P(I_1\cap I_2)]^m P(I_1\cup I_2- I_1\setminus I_2 - I_2 \setminus I_1 )  - [(I_1\cup I_2)P(I_1\cup I_2)]^m  P(I_1 \cup I_2) \nonumber\\
&= F(I_1,I_2; t^m) P (I_1 \cup I_2) + \big[ \big((I_1\cap I_2)P(I_1\cap I_2)\big)^m - (I_1PI_1)^m \big]  P  (I_2\setminus I_1)  \nonumber\\
&+ \big[ \big((I_1\cap I_2)P(I_1\cap I_2)\big)^m - (I_2PI_2)^m \big]  P  (I_1\setminus I_2) \, .
\end{align}
The first summand is in $\mathcal S_1$ by the induction hypothesis. The remaining two summands are of the same form except for the interchange of the indices $1$ and $2$. Thus, it suffices to show that one of them is trace class. To this end, we perform another induction over the statement that
\begin{align}
G_{m'}  \coloneqq \Big[( I_1  (I_1PI_1)^{m'} - ( I_1 \cap I_2)\big((I_1\cap I_2)P(I_1\cap I_2)\big)^{m'} \Big]  P  (I_2\setminus I_1) \in \mathcal S_1
\end{align}
for any $m' \ge 0$. The factors $I_1 \cap I_2$ and $I_1$ we smuggled in only change $G_0$, which is a more convenient induction start. We see that $G_0= (I_1 \setminus I_2 ) P (I_2 \setminus I_1) \in \mathcal S_1$. For the induction step, we consider
\begin{align}
G_{m'+1}&=  \Big[I_1 P  I_1  (I_1PI_1)^{m'} - ( I_1 \cap I_2) P ( I_1 \cap I_2)\big((I_1\cap I_2)P(I_1\cap I_2)\big)^{m'}  \Big]   P (I_2 \setminus I_1)\nonumber \\
&= I_1 P G_{m'} + ( I_1 \setminus I_2) P (I_1 \cap I_2) \big((I_1 \cap I_2 ) P (I_1 \cap I_2 )\big)^{m'} P (I_2 \setminus I_1 ) \, . \label{intersecting intervals ind proof eq1} 
\end{align}
Once more, the first term is trace class by the induction hypothesis and for the second summand, we perform yet another induction. The induction claim is that the operator
\begin{align}
H_{m''} \coloneqq   ( I_1 \setminus I_2) P \left[ P (I_1 \cap I_2 ) P \right] ^{m''} (I_2 \setminus I_1)  \in \mathcal S_1  
\end{align}
for any $m'' \ge 1$. The second summand in \eqref{intersecting intervals ind proof eq1}  is $H_{m'+1}$. For the case $m''=1$, we split the bounded interval $I_1 \cap I_2$ in the middle such that $\overline{I_1 \cap I_2} = \overline {J_1} \cup \overline {J_2}$ for certain (disjoint, bounded) intervals $J_1$ and $J_2$. Secondly, we choose $J_1$ and $J_2$ such that $\operatorname{dist}(I_1 \setminus I_2, J_1 ) = \operatorname{dist} (I_2 \setminus I_1, J_2)= \lvert I_1 \cap I_2 \rvert /2$. Thus, we have 
\begin{align}
\left \lVert H_1 \right \rVert_{\mathcal S_1} = \left \lVert ( I_1 \setminus I_2) P (J_1+J_2) P ( I_2 \setminus I_1) \right \rVert_{\mathcal S_1}  \le \left \lVert (I_1 \setminus I_2) P J_1 \right \rVert_{\mathcal S_1} + \left \lVert J_2 P (I_2 \setminus I_1) \right \rVert_{\mathcal S_1}   < \infty \, .
\end{align}
For the induction step, we introduce the intervals $J_3,J_4$, which are the connected components of $(I_1 \cup I_2 ) ^\complement$. One of them may be empty. Thus, we have $\mathds{1}=(I_1 \cap I_2) + (I_1 \setminus I_2)+ (I_2 \setminus I_1) + J_3 +J_4$. We note that $J_3$ and $J_4$ have positive distance from the finite interval $I_1 \cap I_2$, as $I_1 \setminus I_2$ and $I_2 \setminus I_1$ have positive length. We can now conclude the induction step,
\begin{align}
\left \lVert H_{m''+1} \right \rVert_{\mathcal S_1}  &=  \left \lVert ( I_1 \setminus I_2) P \big(  \mathds{1}- J_3-J_4-(I_1 \setminus I_2)- (I_2 \setminus I_1) \big) P   \left[ P (I_1 \cap I_2 ) P \right] ^{m''} (I_2 \setminus I_1)  \right \rVert_{\mathcal S_1}  \nonumber\\
&\le \left \lVert H_{m''}  \right \rVert_{\mathcal S_1} + \left \lVert  J_3 P (I_1\cap I_2) \right \rVert_{\mathcal S_1} + \left \lVert  J_4 P (I_1\cap I_2) \right \rVert_{\mathcal S_1}\nonumber\\
& + \left \lVert (I_1 \setminus I_2 ) P  H_{m''}  \right \rVert_{\mathcal S_1} + \left \lVert (I_1 \setminus I_2) P (I_2 \setminus I_1)  \right \rVert_{\mathcal S_1} < \infty \, .
\end{align}

In order to compute the trace of $F(I_1,I_2;t^m)$ we resort to the strong continuity of (the indicator functions) $(I_1\setminus I_2)_\eps$ to ${I_1\setminus I_2}$ as $\eps\to0$. If both intervals are bounded, we may assume without loss of generality that $I_1=(0,a)$, $I_2=(1,b)$ with $1<a<b$. Then $I_1\setminus I_2 =(0,1)$ and $I_1\cap I_2 = (1,a)$. For $0<\eps<1$, let $(I_1\setminus I_2)_\eps \coloneqq (0,1-\eps)$. Then, both $\D$-terms in $F$ are trace class and we know how to compute this trace by \autoref{main thm}, 
\begin{align} \tr F(I_1,I_2;t^m) &=\lim_{\eps\to0} \tr\big\{\D\big((I_1\setminus I_2)_\eps,I_2;t^m\big) - \D\big((I_1\setminus I_2)_\eps,I_1\cap I_2;t^m\big)\big\} \nonumber
\\
&= \frac{U(0,1;t^m)}{2\pi^2} \lim_{\eps\to0}\Big\{\ln \Big[\frac{b-1+\eps}{\eps b}\Big] - \ln \Big[\frac{a-1+\eps}{\eps a}\Big]\Big\}\nonumber
\\
&=\frac{U(0,1;t^m)}{2\pi^2} \ln\Big[\frac{a(b-1)}{(a-1)b}\Big]\,,
\end{align}
which is the claimed formula.

Suppose that $I_1$ is bounded and $I_2$ is unbounded. Then, we may assume that $I_1 = (0,a)$, $I_2 = (1,\infty)$ with $1<a$. By using \autoref{main thm} we obtain,
\begin{align} \tr F(I_1,I_2;t^m) &=\lim_{\eps\to0} \tr\big\{\D\big((I_1\setminus I_2)_\eps,I_2;t^m\big) - \D\big((I_1\setminus I_2)_\eps,I_1\cap I_2;t^m\big)\big\} \nonumber
\\
&= \frac{U(0,1;t^m)}{2\pi^2} \lim_{\eps\to0} \Big\{\ln\Big[\frac{1}{\eps}\Big] - \ln\Big[\frac{a-1}{\eps a}\Big]\Big\}\nonumber
\\
&=\frac{U(0,1;t^m)}{2\pi^2} \ln\Big[\frac{a}{a-1}\Big]\,.
\end{align}
This finishes the proof of \autoref{conject} for polynomial test functions.

\end{proof}

\begin{remark} We can generalize the trace-class property of $F(I_1,I_2;f)$ for polynomials $f$ to analytic functions but we did not succeed to deal with more general functions $f$ as in \autoref{assump}, for example by using the Fourier representation \eqref{Fourier rep} of $f$. And we lack an inequality such as \eqref{inequ Sobolev}, which could be applied to find a good bound on the 1-norm of $F(I_1,I_2;f)$ for such functions $f$. We consider this an interesting open problem. For the $\a$-R\'enyi entropy functions $h_\a$ with $\a\le1$ (including the von-Neumann case), we are confident that the trace-class property and the above continuity can be proved using the integral representation \eqref{eq:Ha_integral_representation} and the explicit knowledge of the involved resolvents and thus prove \autoref{conject} for these entropies.
%It is curious that for single intervals $F((0,1),(a,b);f) = F((0,a),(1,b);f)$ if $1<a<b$. 
\end{remark}

\begin{appendix}
\section{Proof of \autoref{theorem:many_intervals}}

We now consider the case of the union of finitely many intervals. First of all, the operator $\D(\mathcal I_1,\mathcal I_2;f)$ is trace class due to inequality \eqref{q inequality}. We just need to replace $I_j$ by $\mathcal I_j$ in the definition of the operator $T$ in \eqref{kernel T} and the constant $N_\b$ from the bound \eqref{T is in Sp eq} accordingly. We introduce $\mathcal I \coloneqq \mathcal I_1\cup\mathcal I_2$. Secondly, we will use the continuity \eqref{continuity} and compute the trace of $\D(\mathcal I_1,\mathcal I_2;f)$ for $f\in\mathsf{C}^2(\R)$ as in \autoref{section 3}. 

Let $I_k:=(a_k,b_k)$, $k\in\{1,\dots,N\}$, be $N\in\N$, $N\ge2$, open intervals with pairwise disjoint closures.
We now generalize the cutoff functions from Section~\ref{section:smooth_Widom} to the
union of $N$ disjoint intervals as $\varphi_{\{1,\dots,N\},\varepsilon}:=\sum_{k=1}^N \varphi_{k,\varepsilon}$. For a single interval $I_k$, we set $\varphi_{\{k\},\varepsilon}:=\varphi_{k,\varepsilon}$. Since the supports of the functions $\varphi_{k,\varepsilon}$, $1\leq k\leq N$, are disjoint,
we get $\varphi^2_{\{1,\dots,N\},\varepsilon}:=\sum_{k=1}^N \varphi^2_{k,\varepsilon}$.

We first prove the following elementary properties of the function $U$ from~\eqref{def U}. 

\begin{lemma}\label{lemma:Uunion}
Let $0<\varepsilon<\frac{1}{2}\min_{1\leq k\leq N}\left\{b_k-a_k\right\}$. Then:
\begin{enumerate}
\item\label{Uunion:point1} for both $\xi_1$ and $\xi_2\in\bigcup_{1\leq k\leq N}I_{k,\varepsilon}$:
\begin{align}
U\big(\varphi^2_{\{1,\dots,N\},\varepsilon}(\xi_1),\varphi^2_{\{1,\dots,N\},\varepsilon}(\xi_2);f\big)=0\,,\nonumber
\end{align}
\item\label{Uunion:point2} for both $\xi_1$ and $\xi_2\in\bigcap_{1\leq k\leq N}I_k^{\complement}$:
\begin{align}
U\big(\varphi_{\{1,\dots,N\},\varepsilon}^2(\xi_1),\varphi_{\{1,\dots,N\},\varepsilon}^2(\xi_2);f\big)=0\,,\nonumber
\end{align}
\item\label{Uunion:point3} for $\xi_1\in\bigcup_{1\leq k\leq N}I_{k,\varepsilon}$ and $\xi_2\in\bigcap_{1\leq k\leq N}I_k^{\complement}$:
\begin{align}
U\big(\varphi_{\{1,\dots,N\},\varepsilon}^2(\xi_1),\varphi_{\{1,\dots,N\},\varepsilon}^2(\xi_2);f\big)=U(1,0;f)\,,\nonumber
\end{align}
\item\label{Uunion:point4} for $\xi_1\in\bigcap_{1\leq k\leq N}I_k^{\complement}$ and $\xi_2\in\bigcup_{1\leq k\leq N}I_{k,\varepsilon}$:
\begin{align}
U\big(\varphi_{\{1,\dots,N\},\varepsilon}^2(\xi_1),\varphi_{\{1,\dots,N\},\varepsilon}^2(\xi_2);f\big)=U(0,1;f)\,.\nonumber
\end{align}
\end{enumerate}
\end{lemma}
\begin{proof}
Without loss of generality, suppose that $b_\ell<a_{\ell+1}$ for $1\leq \ell\leq N-1$.

Case~\eqref{Uunion:point1}: There exist two (possibly identical) indices $1\leq \ell,m\leq N$
such that $\xi_1\in I_{\ell,\varepsilon}\subset I_\ell$ and $\xi_2\in I_{m,\varepsilon}\subset I_m$.
Therefore $\varphi_{\{1,\dots,N\},\varepsilon}(\xi_1)=\varphi_{\ell,\varepsilon}(\xi_1)$ and
$\varphi_{\{1,\dots,N\},\varepsilon}(\xi_2)=\varphi_{m,\varepsilon}(\xi_2)$ and this yields:
\begin{align}
U\big(\varphi^2_{\{1,\dots,N\},\varepsilon}(\xi_1),\varphi^2_{\{1,\dots,N\},\varepsilon}(\xi_2);f\big)
=U\big(\varphi^2_{\ell,\varepsilon}(\xi_1),\varphi^2_{m,\varepsilon}(\xi_2);f\big)=U(1,1;f)=0\,.\nonumber
\end{align}

Case~\eqref{Uunion:point2}: For this choice of the variables, $\xi_1,\xi_2\not\in I_k$ for all $k\in\{1,\dots,N\}$. Therefore:
\begin{align}
U\big(\varphi^2_{\{1,\dots,N\},\varepsilon}(\xi_1),\varphi^2_{\{1,\dots,N\},\varepsilon}(\xi_2);f\big)
=U(0,0;f)=0\,.\nonumber
\end{align}

Case~\eqref{Uunion:point3}: We combine here the cases~\eqref{Uunion:point1} and~\eqref{Uunion:point2} above. 
Then there exists an index $1\leq \ell\leq N$ such that $\xi_1\in I_{\ell,\varepsilon}\subset I_\ell$, which yields:
\begin{align}
U\big(\varphi^2_{\{1,\dots,N\},\varepsilon}(\xi_1),\varphi^2_{\{1,\dots,N\},\varepsilon}(\xi_2);f\big)
=U(\varphi^2_{\ell,\varepsilon}(\xi_1),0;f)=U(1,0;f)\,.\nonumber
\end{align}

Case~\eqref{Uunion:point4} follows analogously to case~\eqref{Uunion:point3}.
\end{proof}

We may now extend~\autoref{main thm} to the case of finitely many intervals. 

\begin{proof}[Proof of \autoref{theorem:many_intervals}] We do not need to repeat the proof for general test functions $f$ satisfying the \autoref{assump} and we can immediately turn to the smooth case. All we have to do is to prove the generalization of \autoref{trace smooth} to the present situation.

We define for fixed smooth $f$
\begin{align}\label{eq:DeltaU}
\Delta U(\xi_1,\xi_2,\varepsilon):&=
U\big(\varphi^2_{\mathcal{P}_1,\varepsilon}(\xi_1),\varphi^2_{\mathcal{P}_1,\varepsilon}(\xi_2);f\big)
+U\big(\varphi^2_{\mathcal{P}_2,\varepsilon}(\xi_1),\varphi^2_{\mathcal{P}_2,\varepsilon}(\xi_2);f\big)\nonumber\\
&-U\big(\varphi^2_{\{1,\dots,N\},\varepsilon}(\xi_1),\varphi^2_{\{1,\dots,N\},\varepsilon}(\xi_2);f\big)\,,
\end{align}
where $\varphi^2_{\mathcal{P}_j,\varepsilon}\coloneqq \sum_{k\in\mathcal P_j} \varphi^2_{k,\eps}$ for $j=1,2$. Then we will analyze the integral
\begin{align}\label{eq:DeltaU_trace_integral}
\mathcal B(\varphi^2_{\mathcal{P}_1,\eps};f) + \mathcal B(\varphi^2_{\mathcal{P}_2,\eps};f)
-\mathcal B(\varphi^2_{\{1,\dots,N\},\eps};f)=\frac{1}{8\pi^2}\int_{\R}\mathrm{d}\xi_1\int_{\R}\mathrm{d}\xi_2\,
\frac{\Delta U(\xi_1,\xi_2,\varepsilon)}{|\xi_1-\xi_2|^2}\,.
\end{align}

Since the intervals $I_k$, $1\leq k\leq N$, are pairwise disjoint, it follows that every
interval $I_i$, $i\in\mathcal{P}_1$, that belongs to subsystem $\calI_1$ is a subset of $I_j^{\complement}$,
for every $j\in\mathcal{P}_2$. This yields the inclusion
\begin{align}\label{eq:inclusion1}
\calI_{1,\varepsilon}:=\bigcup_{i\in\mathcal{P}_1}I_{i,\varepsilon}&\subset\bigcup_{i\in\mathcal{P}_1}I_i
\subset\bigcap_{j\in\mathcal{P}_2}I_j^{\complement}
=\Biggl(\bigcup_{j\in\mathcal{P}_2}I_j\Biggr)^{\complement}=:\calI_2^{\complement}\,,
\end{align}
as well as
\begin{align}\label{eq:inclusion2}
\calI_{2,\varepsilon}\subset\calI_1^{\complement}\,.
\end{align}
Moreover
\begin{align}\label{eq:inclusion3}
\calI_{1,\varepsilon},\calI_{2,\varepsilon}\subset
\bigcup_{j\in\{1,\dots,N\}}I_{j,\varepsilon}=:\calI_{\varepsilon}\,,
\end{align}
since $\mathcal{P}_1\cup\mathcal{P}_2=\{1,\dots,N\}$, and
\begin{align}\label{eq:inclusion4}
\calI^{\complement}&\subset\calI_1^{\complement},\calI_2^{\complement}\,.
\end{align}

We now divide the real line $\R$ in $4N+1$ disjoint subsets according to the partition
\begin{align}
{\mathcal{P}}_{\R}:=&\big\{
(-\infty,a_1), [a_1,a_1+\varepsilon], (a_1+\varepsilon,b_1-\varepsilon),
{[b_1-\varepsilon,b_1]},\nonumber\\
&\hspace{0.5cm} (b_1,a_2),[a_2,a_2+\varepsilon],
(a_2+\varepsilon,b_2-\varepsilon), [b_2-\varepsilon,b_2],\nonumber\\
&\hspace{0.5cm}\dots,\nonumber\\
&\hspace{0.5cm} (b_{N-1},a_N),[a_N,a_N+\varepsilon], (a_N+\varepsilon,b_N-\varepsilon), [b_N-\varepsilon,b_N],(b_N,+\infty)\big\}\,,\nonumber
\end{align}
and we divide the integration domain $\R^2$ according to the finite partition
${\mathcal{P}}_{\R^2}:=\{A\times B:A,B\in{\mathcal{P}}_{\R}\}$. We now consider the contributions
to~\eqref{eq:DeltaU_trace_integral} that arise from each element of ${\mathcal{P}}_{\R^2}$.

We repeatedly employ Lemma~\eqref{lemma:Uunion} and the set inclusions~\eqref{eq:inclusion1},
\eqref{eq:inclusion2}, \eqref{eq:inclusion3}, and~\eqref{eq:inclusion4}. Then, we see that whenever
we consider the following combinations of the variables $\xi_1,\xi_2$,
the function $\Delta U$ takes (finite) constant values independent of $\varepsilon$:
\begin{enumerate}
\item\label{listJ1J2_case1} if both $\xi_1$ and $\xi_2\in\calI_{1,\varepsilon}$, then
$\Delta U(\xi_1,\xi_2,\varepsilon)=0$;

\item\label{listJ1J2_case2} if both $\xi_1$ and $\xi_2\in\calI^{\complement}$, then $\Delta U(\xi_1,\xi_2,\varepsilon)=0$;

\item\label{listJ1J2_case3} if $\xi_1\in\calI_{1,\varepsilon}$ and
$\xi_2\in\calI_{2,\varepsilon}$, then $\Delta U(\xi_1,\xi_2,\varepsilon)=2U(1,0;f)$;

\item\label{listJ1J2_case4} if $\xi_1\in\calI_{1,\varepsilon}$ and $\xi_2\in\calI^{\complement}$, then 
$\Delta U(\xi_1,\xi_2,\varepsilon)=0$;

\item\label{listJ1J2_case5} if both $\xi_1$ and $\xi_2\in\bigcup_{k\in\mathcal{P}_1}[a_k,a_k+\varepsilon]
\cup [b_k-\varepsilon,b_k]$, then $\Delta U(\xi_1,\xi_2,\varepsilon)=0$;

\item\label{listJ1J2_case6} if $\xi_1\in\bigcup_{k\in\mathcal{P}_1}[a_k,a_k+\varepsilon]
\cup [b_k-\varepsilon,b_k]$ and $\xi_2\in\calI_{1,\varepsilon}$, then $\Delta U(\xi_1,\xi_2,\varepsilon)=0$;

\item\label{listJ1J2_case7} if $\xi_1\in\bigcup_{k\in\mathcal{P}_1}[a_k,a_k+\varepsilon]\cup [b_k-\varepsilon,b_k]$ and $\xi_2\in\calI^{\complement}$, then $\Delta U(\xi_1,\xi_2,\varepsilon)=0$.
\end{enumerate}

Identical results follows from the list above if we interchange the set indices 1 and 2.

We now focus on the contribution to~\eqref{eq:DeltaU_trace_integral} on
the remaining elements of ${\mathcal{P}}_{\R^2}$ not treated above, which, in general, explicitly depend on
the parameter $\varepsilon$. We denote a generic element of $\mathcal{P}_{\R^2}$ that falls into
this group by $A^{\varepsilon}\times B^{\varepsilon}$. The latter is characterized
by the properties that the Lebesgue measures
\begin{align}
|A^{\varepsilon}|<\infty,|B^{\varepsilon}|<\infty\quad \mbox{ for }\quad
0<\varepsilon<\varepsilon_0\coloneqq \frac{1}{2}\min_{1\leq k\leq N}\{b_k-a_k\}\,,\nonumber
\end{align}
as well as
\begin{align}\label{eq:limAB}
\lim_{\varepsilon\rightarrow 0}|A^{\varepsilon}|=0
\hspace{0.5cm}\lor\hspace{0.5cm}
\lim_{\varepsilon\rightarrow 0}|B^{\varepsilon}|=0\,.
\end{align}

We note that all elements of ${\mathcal{P}}_{\R^2}$ that are not bounded have already been
addressed in the cases~\eqref{listJ1J2_case2}, \eqref{listJ1J2_case4} and~\eqref{listJ1J2_case7} in the list above. Therefore the term $\xi_1-\xi_2$ in the denominator of~\eqref{eq:DeltaU_trace_integral}
remains always finite, and it is bounded by $|\xi_1-\xi_2|\leq b_N-a_1$.
Furthermore, $\xi_1-\xi_2$ never tends to 0\ in the $\varepsilon\rightarrow 0$ limit,
since this possibility has been ruled out by the cases~\eqref{listJ1J2_case1}, \eqref{listJ1J2_case2},
\eqref{listJ1J2_case4}, \eqref{listJ1J2_case5}, \eqref{listJ1J2_case6} and~\eqref{listJ1J2_case7}
in the list above.
A lower bound reads $|\xi_1-\xi_2|\geq\min_{1\leq k\leq N-1}\{a_{k+1}-b_k\}$, which
is always finite since the closures of the intervals are pairwise disjoint.

As a consequence, the remaining contributions to~\eqref{eq:DeltaU_trace_integral}
are given by proper integrals bounded below and above by the inequalities 
\begin{align}
D_m\int_{A^{\varepsilon}}\mathrm{d}\xi_1
\int_{B^{\varepsilon}}\mathrm{d}\xi_2\,\Delta U(\xi_1,\xi_2,\varepsilon)
&\leq\int_{A^{\varepsilon}}\mathrm{d}\xi_1\int_{B^{\varepsilon}}\mathrm{d}\xi_2\,
\frac{\Delta U(\xi_1,\xi_2,\varepsilon)}{|\xi_1-\xi_2|^2}\nonumber\\
&\leq D_M\int_{A^{\varepsilon}}\mathrm{d}\xi_1
\int_{B^{\varepsilon}}\mathrm{d}\xi_2\,\Delta U(\xi_1,\xi_2,\varepsilon)\,,\nonumber
\end{align}
where $D_m\coloneqq\frac{1}{(b_N-a_1)^2}$ and $D_M\coloneqq\frac{1}{(\min_{1\leq k\leq N-1}\{a_{k+1}-b_k\})^2}$. This yields
\begin{align}\label{eq:DeltaU_integral_vanish}
\lim_{\varepsilon\rightarrow 0}\int_{A^{\varepsilon}}\mathrm{d}\xi_1
\int_{B^{\varepsilon}}\mathrm{d}\xi_2\,
\frac{\Delta U(\xi_1,\xi_2,\varepsilon)}{|\xi_1-\xi_2|^2}
=0\,,
\end{align}
since by~\eqref{eq:limAB} either $A^{\varepsilon}$ or $B^{\varepsilon}$ is of vanishing measure.

Summarizing, the only non-zero contributions to the integral~\eqref{eq:DeltaU_trace_integral}
in the $\varepsilon\rightarrow 0$ limit come from the case~\eqref{listJ1J2_case3} in the list above,
for $\xi_1\in\calI_{1,\varepsilon}$ and
$\xi_2\in\calI_{2,\varepsilon}$ or interchanging the set indices 1 and 2.
Altogether, we get:
\begin{align}\label{D_generalized}
&\tr\D(\calI_1,\calI_2;f)=
\lim_{\varepsilon\rightarrow 0}\Big[\mathcal B(\varphi^2_{\mathcal{P}_1,\eps};f)
+\mathcal B(\varphi^2_{\mathcal{P}_2,\eps};f)-\mathcal B(\varphi^2_{\{1,\dots,N\},\eps};f)\Big]\nonumber\\
&=\frac{1}{8\pi^2}\sum_{k\in\mathcal{P}_1}\sum_{\ell\in\mathcal{P}_2}\lim_{\varepsilon\rightarrow 0}
\left(\int_{a_k+\varepsilon}^{b_k-\varepsilon}\mathrm{d}\xi_1
\int_{a_\ell+\varepsilon}^{b_\ell-\varepsilon}\mathrm{d}\xi_2\,
\frac{\Delta U(\xi_1,\xi_2,\varepsilon)}{(\xi_2-\xi_1)^2}
+\int_{a_\ell+\varepsilon}^{b_\ell-\varepsilon}\mathrm{d}\xi_1
\int_{a_k+\varepsilon}^{b_k-\varepsilon}\mathrm{d}\xi_2\,
\frac{\Delta U(\xi_1,\xi_2,\varepsilon)}{(\xi_1-\xi_2)^2}
\right)\nonumber\\
&=\frac{U(0,1;f)}{2\pi^2}\sum_{k\in\mathcal{P}_1}\sum_{\ell\in\mathcal{P}_2}\lim_{\varepsilon\rightarrow 0}
\int_{a_k+\varepsilon}^{b_k-\varepsilon}\mathrm{d}\xi_1
\int_{a_\ell+\varepsilon}^{b_\ell-\varepsilon}\mathrm{d}\xi_2\,
\frac{1}{(\xi_2-\xi_1)^2}\nonumber\\
&=\frac{U(0,1;f)}{2\pi^2}\sum_{k\in\mathcal{P}_1}\sum_{\ell\in\mathcal{P}_2}\lim_{\varepsilon\rightarrow 0}
\int_{a_k+\varepsilon}^{b_k-\varepsilon}\mathrm{d}\xi_1
\left(\frac{1}{a_\ell+\varepsilon-\xi_1}-\frac{1}{b_\ell-\varepsilon-\xi_1}\right)\nonumber\\
&=\frac{U(0,1;f)}{2\pi^2}\sum_{k\in\mathcal{P}_1}\sum_{\ell\in\mathcal{P}_2}\lim_{\varepsilon\rightarrow 0}
\ln\Big[\frac{|a_k-a_\ell||b_k-b_\ell|}
{|a_\ell-b_k+2\varepsilon||b_\ell-a_k-2\varepsilon|}\Big]\nonumber\\
&=\frac{U(0,1;f)}{2\pi^2}\sum_{k\in\mathcal{P}_1}\sum_{\ell\in\mathcal{P}_2}
\ln\Big[\frac{|a_k-a_\ell||b_k-b_\ell|}{|a_k-b_\ell||b_k-a_\ell|}\Big]\,.
\end{align}
Proceeding as in the proof of~\autoref{main thm}, the result~\eqref{D_generalized} yields the claim.
\end{proof}

\end{appendix}

\end{document}